\RequirePackage{fix-cm}
\documentclass[smallextended]{svjour3}       
\smartqed  
\usepackage{graphicx}
\usepackage{amsmath,amssymb,array,graphicx,verbatim}
\newtheorem{assumption}{Assumption}
\newtheorem{prop}{Property}
\usepackage{enumerate}
\usepackage[sort&compress]{natbib}

\begin{document}

\title{Signaling for Decentralized Routing in a Queueing Network
\thanks{Preliminary versions of this paper appeared in \cite{ouyang2013} and \cite{ouyang2014}. This work was supported in part by National Science Foundation (NSF) Grant CCF-1111061 and NASA grant NNX12AO54G.}
}


\author{Yi~Ouyang \and Demosthenis~Teneketzis
}


\institute{Y. Ouyang \at
              Department of Electrical Engineering and Computer Science, University of Michigan, Ann Arbor, MI\\
              \email{ouyangyi@umich.edu}           
           \and
           D. Teneketzis \at       
           	  Department of Electrical Engineering and Computer Science, University of Michigan, Ann Arbor, MI\\      
              \email{teneket@umich.edu}    
}

\date{Received: date / Accepted: date}

\maketitle

\begin{abstract}
A discrete-time decentralized routing problem in a service system consisting of two service stations and two controllers is investigated. Each controller is affiliated with one station.
Each station has an infinite size buffer. Exogenous customer arrivals at each station occur with rate $\lambda$.
Service times at each station have rate $\mu$. At any time, a controller can route one of the customers waiting in its own station to the other station. Each controller knows perfectly the queue length in its own station and observes the exogenous arrivals to its own station as well as the arrivals of customers sent from the other station. 
At the beginning, each controller has a probability mass function (PMF) on the number of customers in the other station.
These PMFs are common knowledge between the two controllers.
At each time a holding cost is incurred at each station due to the customers waiting at that station.
The objective is to determine routing policies for the two controllers that minimize either the total expected holding cost over a finite horizon or the average cost per unit time over an infinite horizon.
In this problem there is implicit communication between the two controllers; whenever a controller decides to send or not to send a customer from its own station to the other station it communicates information about its queue length to the other station.
This implicit communication through control actions is referred to as signaling in decentralized control.
Signaling results in complex communication and decision problems. In spite of the complexity of signaling involved, it is shown that an optimal signaling strategy is described by a threshold policy which depends on the common information between the two controllers; this threshold policy is explicitly determined.
\end{abstract}

\section{Introduction}
\label{sec:intro}
Routing problems to parallel queues arise in many modern technological systems such as communication, transportation and sensor networks.
The majority of the literature on optimal routing in parallel queues addresses situations where the information is centralized, either perfect (see \cite{winston1977optimality,weber1978optimal,davis1977optimal,ephremides1980simple,lin1984optimal,hajek1984optimal,whitt1986deciding,weber1987optimal,hordijk1990optimality,hordijk1992assignment,menich1991optimality,foley2001join,akgun2012understanding} and references therein)
or imperfect (see \cite{beutler1989routing,kuri1995optimal} and references therein).
Very few results on optimal routing to parallel queues under decentralized information are currently available.
The authors of \cite{cogill2006approximate} present a heuristic approach to decentralized routing in parallel queues.
In (\cite{ying2011throughput,reddy2012distributed,manfredi2014decentralized,abdollahi2008novel,si2013decentralized} and references therein), decentralized routing policies that stablize queueing networks are considered.
The work in \cite{pandelis1996simple} presents an optimal policy to a routing problem with a one-unit delay sharing information structure.

In this paper we investigate a decentralized routing problem in discrete time.
We consider a system consisting of two service stations/queues, called $Q_1$ and $Q_2$ and two controllers, called $C_1$ and $C_2$.
Controller $C_1$ (resp. $C_2$) is affiliated with service station $Q_1$ (resp. $Q_2$). 
Each station has an infinite size buffer. 
The processes describing exogenous customer arrivals at each station are independent Bernoulli with parameter $(\lambda)$. 
The random variables describing the service times at each station are independent geometric with parameter $(\mu)$. 
At any time each controller can route one of the customers waiting in its own queue to the other station.
Each controller knows perfectly the queue length in its own station, and observes the exogenous arrivals in its own station as well as the arrivals of customers sent from the other station.
At the beginning, controller $C_1$ (resp. $C_2$) has a probability mass function (PMF) on the number of customers in station $Q_2$ (resp. $Q_1$).
These PMFs are common knowledge between the controllers.
At each time a holding cost is incurred at each station due to the customers waiting at that station.
The objective is to determine decentralized routing policies for the two controllers that minimize either
the total expected holding cost over a finite horizon or the average cost per unit time over an infinite horizon.
Preliminary versions of this paper appeared in \cite{ouyang2013} (for the finite horizon problem) and \cite{ouyang2014} (for the infinite horizon average cost per unit time problem).

In the above described routing problem, each controller has different information.
Furthermore, the control actions/routing decisions of one controller affect the information of the other controller.
Thus, the information structure of this decentralized routing problem is non-classical with control sharing (see \cite{mahajan2012optimal} for non-classical control sharing information structures).
Non-classical information structures result in challenging signaling problems (see \cite{ho1980team}).
Signaling occurs through the routing decisions of the controllers. 
Signaling is, in essence, a real-time encoding/communication problem within the context of a decision making problem.
By sending or not sending a customer from $Q_1$ (resp. $Q_2$) to $Q_2$ (resp. $Q_1$) controller $C_1$ (resp. $C_2$) communicates at each time instant a compressed version of its queue length to $C_2$ (resp. $C_1$). 
For example, by sending a customer from $Q_1$ to $Q_2$ at time $t$, $C_1$ may signal to $C_2$ that $Q_1$'s queue length is above a pre-specified threshold $l_t$. This information allows $C_2$ to have a better estimate of $Q_1$'s queue length and, therefore, make better routing decisions about the customers in its own queue; the same arguments hold for the signals send (through routing decisions) from $C_2$ to $C_1$.
Thus, signaling through routing decisions has a triple function: communication, estimation and control.

Within the context  of the problems described above, 
there is enormous number of signaling possibilities.
For example, there is an arbitrarily large number of choices of the sequences of pre-specified thresholds $l_1,l_2,\dots,l_t,\dots$ and these choices are only a small subset of all the possible sequences of binary partitions of the set of non-negative integers that describe all choices available to $C_1$ and $C_2$.
All these possibilities result in highly non-trivial decision making problems.
It is the presence of signaling that distinguishes the problem formulated in this paper from all other routing problems in parallel queues investigated so far. 

Some basic questions associated with the analysis of this problem are:

What is an information state (sufficient statistic) for each controller? 
How is signaling incorporated in the evolution/update of the information state?
Is there an explicit description of an optimal signaling strategy? 
We will answer these questions in Section \ref{sec:SR}-\ref{sec:infinite} and will discuss them further in Section \ref{sec:diss}.

\subsection*{Contribution of the paper}
The signaling feature of our problem distinguishes it from all previous routing problems in parallel queues.
In spite of the complexity of signaling, we show that an optimal decentralized strategy is described by a single threshold routing policy where the threshold depends on the common information between the two controllers.
We explicitly determine this threshold via simple computations.


\subsection*{Organization}
The rest of the paper is organized as follows.
In Section \ref{sec:model} we present the model for the queueing system and formulate the finite horizon and infinite horizon decentralized routing problems. 
In Section \ref{sec:SR} we present structural results for optimal policies.
In Section \ref{sec:prelim} we present a specific decentralized routing policy, which we call $\hat{g}$, and state some features associated with its performance.
In Section \ref{sec:finite}, we show that when the initial queue lengths in $Q_1$ and $Q_2$ are equal, $\hat{g}$ is an optimal policy for the finite horizon decentralized routing problem.
In Section \ref{sec:infinite}, we show that $\hat{g}$ is an optimal decentralized routing policy for the infinite horizon average cost per unit time problem.  We conclude in Section \ref{sec:diss}.

\subsection*{Notation}
Random variables (r.v.s) are denoted by upper case letters, their realization by the corresponding lower case letter.	
In general, subscripts are used as time index while superscripts are used to index service stations. 
For time indices $t_1\leq t_2$, $X_{t_1:t_2}$ is the short hand notation for $(X_{t_1},X_{t_1+1},...,X_{t_2})$.
For a policy $g$, we use $X^{g}$ to denote that the r.v. $X^{g}$ depends on the choice of policy $g$.
We use vectors in $\mathbb{R}^{\mathbb{Z}_+} $ to denote PMFs (Probability Mass Functions,) where $\mathbb{Z}_+$ denotes the set of non-negative integers.
We also use a constant in $\mathbb{Z}_+$ to denote the corner PMF that represents a constant r.v..
i.e. a constant $c \in \mathbb{Z}_+$ denotes the PMF whose entries are all zero except the $c$th.

\section{System Model and Problem Formulation}
\label{sec:model}
\subsection*{System Model}
The queueing/service system shown in Figure \ref{fig:system}, operates in discrete time.
\begin{figure}
\centering
\includegraphics[scale=0.4]{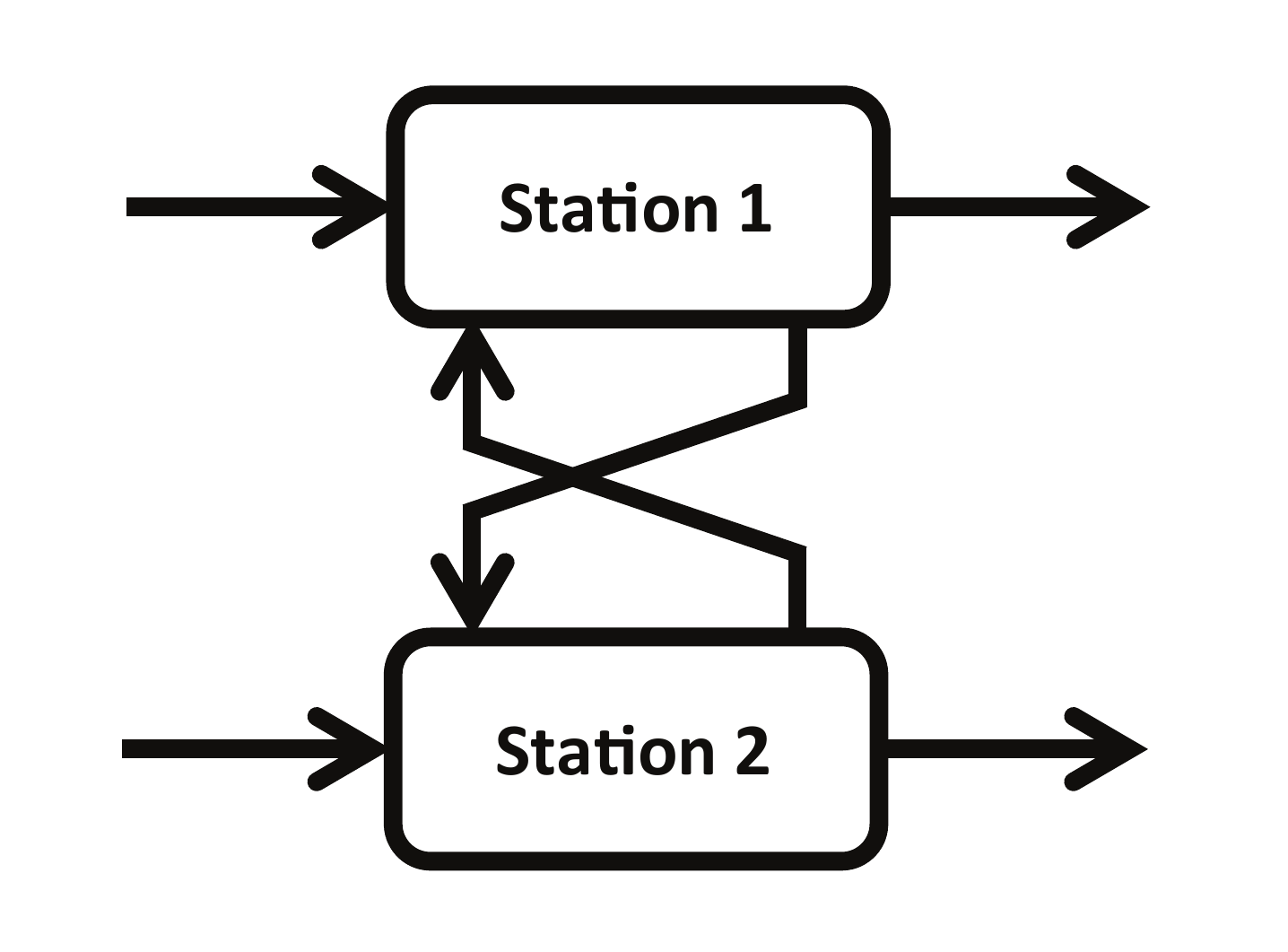}
\caption{The Queueing System }
\label{fig:system}
\end{figure}
The system consists of two service stations/queues, $Q_1$ and $Q_2$ with infinite size buffers.
Controllers $C_1$ and $C_2$ are affiliated with queues $Q_1$ and $Q_2$, respectively.
Let $X^i_t$ denote the number of customers waiting, or in service, in $Q_i,i=1,2$, at the beginning of time $t$.
Exogenous customer arrivals at $Q_i,i=1,2$, occur according to a Bernoulli process $\{A^i_t,t\in\mathbb{Z}_+\}$ with parameter $\lambda$. 
Service times of customers at $Q_i,i=1,2$ are described by geometric random variables with parameter $\mu$.
We define a Bernoulli process $\{D^i_t,t\in\mathbb{Z}_+\}$ with parameter $\mu$. Then $\{D^i_t1_{\{X^i_t\neq 0\}},t\in\mathbb{Z}_+\}$ 
describes the customer departure process from $Q_i,i=1,2$.
At any time $t$, a controller can route one of the customers in its own queue to the other queue.
Let $U^i_t$ denote the routing decision of controller $C_i$ at $t$ ($i=1,2$); if $U^i_t=1$ (resp. $0$) one customer (resp. no customer) is routed from $Q_i$ to $Q_j$ ($j\neq i$).
At any time $t$, each controller $C_i,i=1,2$, knows perfectly the number of customers $X^i_{0:t},i=1,2$, in its own queue; furthermore, it observes perfectly the arrival stream $A^i_{0:t}$ to its own queue, and the arrivals due to customers routed to its queue from the other service station up to time $t-1$, i.e. $U^j_{0:t-1}, j\neq i$.
The order of arrivals $A^i_t$, departures $D^i_t$ and controller decisions $U^i_t$ concerning the routing of customers from one queue to the other is shown in Figure \ref{fig:timeline}. The dynamic evolution of the number of customers 
$X^i_t,i=1,2$ is described by
\begin{align}
&X^1_{t+1} = \overline{X}^1_t-U^1_t+U^2_t, \label{Model:dynamic1}\\
&X^2_{t+1} = \overline{X}^2_t-U^2_t+U^1_t,\label{Model:dynamic2}
\end{align}
where for $i=1,2$,
\begin{align}
&\overline{X}^i_t = \left(X^i_t - D^i_t \right)^++A^i_t,\label{Model:dynamic3}
\end{align}
and $(x)^+ := \max(0,x)$.
We assume that the initial queue lengths $X^1_0,X^2_0$ and the processes $\{A^1_t,t \in \mathbb{Z}_+\}$, $\{A^2_t,t \in \mathbb{Z}_+\}$, $\{D^1_t,t \in \mathbb{Z}_+\}$, $\{D^2_t,t \in \mathbb{Z}_+\}$ are mutually independent
and their distributions are known by both controllers $C_1$ and $C_2$.
Let $\pi^1_0$ and $\pi^2_0$ be the PMFs on the initial queue lengths $X^1_0,X^2_0$, respectively.
At the beginning of time $t=0$, $C_1$ (resp. $C_2$) knows $X^1_0$ (resp. $X^2_0$).
Furthermore $C_1$'s (resp. $C_2$'s) knowledge of the queue length $X^2_0$ (resp. $X^1_0$) at the other station is described by the PMF $\pi^2_0$ (resp. $\pi^1_0$).
The information of controller $C_i,i=1,2$, at the moment it makes the decision $U^i_t, t=0,1,\dots$, is 
\begin{align}
I^i_t := \left\lbrace X^i_{0:t},A^i_{0:t},\overline{X}^i_{0:t},U^1_{0:t-1},U^2_{0:t-1},\pi^1_0,\pi^2_0\right\rbrace,i=1,2.
\end{align}
The controllers' routing decisions/control actions $U^i_t$ are generated according to 
\begin{align}
U^i_t = g^i_t\left(I^i_t\right), i=1,2, t\in\mathbb{Z}_+, \label{Model:defgi1}
\end{align}
where
\begin{align}
g^i_t:& (\mathbb{Z}_+)^{t+1}\times \{0,1\}^{t+1}\times (\mathbb{Z}_+)^{t+1}\times \{0,1\}^{t}\times \nonumber\\ 
&\times \{0,1\}^{t}\times\mathbb{R}^{\mathbb{Z}_+}\times\mathbb{R}^{\mathbb{Z}_+} \mapsto \mathcal{U}^i_t.\label{Model:defgi2}
\end{align}
The control action space $\mathcal{U}^i_t$ at time $t$ depends on $\overline{X}^i_t$. Specifically
\begin{align}
\mathcal{U}^i_t = \left\lbrace
\begin{array}{ll}
\{0\} & \text{when } \overline{X}^i_t=0,\\
\{0,1\}  & \text{otherwise}.
\end{array}
\right. 
\label{Model:defUspace}
\end{align}
Define $\mathcal{G}_d$ to be the set of feasible decentralized routing policies; that is
\begin{align}
\mathcal{G}_d =
\{(g^1,g^2):\quad & g^i = (g^i_0,g^i_1,\dots,g^i_t,\dots),i=1,2 \nonumber\\
		&\text{and } g^i_t \text{ is of form given by (\ref{Model:defgi1})-(\ref{Model:defgi2})}\}.
\label{Model:decpolicies}
\end{align}

\begin{figure}
\centering
\includegraphics[scale=0.5]{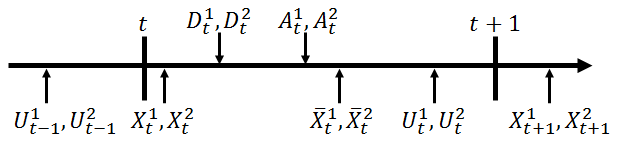}
\caption{The order of variables}
\label{fig:timeline}
\end{figure}

We study the operation of the system defined in this section, first over a finite horizon, then over an infinite horizon.

\subsection{The finite horizon problem} \label{sub:model:finite}
For the problem with a finite horizon $T$, we assume the holding cost incurred by the customers present in $Q_i$ at time $t=0,1,\dots,T-1$ is $c_t(X^i_t), i=1,2$, where $c_t(\cdot)$ is a convex and increasing function. 
Then, the objective is to determine decentralized routing policies $g \in \mathcal{G}_d$
so as to minimize
\begin{align}
J^g_T(\pi^1_0,\pi^2_0) := \mathbf{E}\left[\left.\sum_{t=0}^{T-1}\left(c_t\left(X^{1,g}_t\right)+c_t\left(X^{2,g}_t\right)\right)
\right |\pi^1_0,\pi^2_0\right]  \label{Model:finite:cost}
\footnote{1}
\end{align}
for any PMFs $\pi^1_0,\pi^2_0$ on the initial queue lengths.
\footnotetext[1]{The expectation in all equations appearing in this paper is with respect to the probability measure induced by the policy $g \in \mathcal{G}_d$.}

\subsection{The infinite horizon average cost per unit time problem}\label{sub:model:infinite}
For the infinite horizon average cost per unit time problem, we assume the holding cost incurred by the customers present in $Q_i$ at each time is 
a convex and increasing function $c_t(\cdot):=c(\cdot), i=1,2$. 
Then, the objective is to determine decentralized routing policies $g=(g^1,g^2)\in \mathcal{G}_d$
so as to minimize
\begin{align}
&J^g(\pi^1_0,\pi^2_0) \nonumber\\
:= &\limsup_{T \rightarrow \infty}
\frac{1}{T} J^g_T(\pi^1_0,\pi^2_0)
\nonumber\\
=& \limsup_{T \rightarrow \infty}
\frac{1}{T} \mathbf{E}\left[\left.\sum_{t=0}^{T-1}\left(c\left(X^{1,g}_t\right)+c\left(X^{2,g}_t\right)\right)\right|\pi^1_0,\pi^2_0\right] \label{Model:infinite:cost}
\end{align}
for any PMFs $\pi^1_0,\pi^2_0$ on the initial queue lengths.


\section{Qualitative Properties of Optimal Policies}
\label{sec:SR}
In this section we present a qualitative property of an optimal routing policy for both the finite horizon and the infinite horizon problem.
For that matter we first introduce the following notation.

We denote by $\Pi^1_t$ and $\Pi^2_t$ the PMFs on $X^1_t$ and $X^2_t$, respectively, conditional on all previous decisions $\{U^1_{0:t-1},U^2_{0:t-1}\}$. $\Pi^i_t, i=1,2$ is defined by
\begin{align}
\Pi^i_t(x) := \mathbf{P}\left( X^i_t = x |U^1_{0:t-1},U^2_{0:t-1} \right), x \in \mathbb{Z}_+.
\label{eq:Pit}
\end{align}
Similarly, we define the conditional PMFs $\overline{\Pi}^1_t$, $\overline{\Pi}^2_t$ on $\overline{X}^1_t$ and $\overline{X}^2_t$, respectively, as follows.
\begin{align}
&\overline{\Pi}^i_t(x) := \mathbf{P}\left( \overline{X}^i_t = x |U^1_{0:t-1},U^2_{0:t-1} \right),i = 1,2, x \in \mathbb{Z}_+.
\label{eq:Pitbar} 
\end{align}
Note that for any policy $g \in \mathcal{G}_d$ all the above defined PMFs are functions of $\{U^1_{0:t-1},U^2_{0:t-1}\}$.
Since both controllers $C_1$ and $C_2$ know $\{U^1_{0:t-1},U^2_{0:t-1}\}$ at time $t$, 
the PMFs defined by (\ref{eq:Pit})-(\ref{eq:Pitbar}) are common knowledge \cite{aumann1976agreeing} between the controllers.

We take $\overline{X}^i_t,i=1,2$, to be station $Q_i$'s state at time $t$.
Combining (\ref{Model:dynamic1})-(\ref{Model:dynamic3}) we obtain, for $i=1,2,$
\begin{align}
\overline{X}^i_{t+1}=&\left(\overline{X}^i_t-U^i_t+U^j_t -D^i_{t+1}\right)^++A^i_{t+1}\nonumber\\
	     :=&f^i_t\left(\overline{X}^i_t,U^i_t,U^j_t,W^i_t\right),
	     \label{SR:eqdynamic}
\end{align}
where the random variables $W^i_t:=(A^i_{t+1},D^i_{t+1}),i=1,2,t=0,1,\dots$ are mutually independent.
\\
The holding cost at time $t,t=0,1,\dots$ can be written as
\begin{align}
&\rho_t\left(\overline{X}^1_{t},\overline{X}^2_{t},U^1_{t},U^2_{t}\right) \nonumber\\
:=& c_{t+1}\left(\overline{X}^1_{t}-U^1_{t}+U^2_t\right)+c_{t+1}\left(\overline{X}^2_{t}-U^2_{t}+U^1_t\right)\nonumber\\
=& c_{t+1}\left(X^1_{t+1}\right)+c_{t+1}\left(X^2_{t+1}\right).\label{SR:cost}
\end{align}
Note that for any time horizon $T$ the total expected holding cost due to (\ref{SR:cost}) is equivalent to the total expected holding cost defined by (\ref{Model:finite:cost}) since for any policy $g \in \mathcal{G}_d$
\begin{align}
&J^g_T(\pi^1_0,\pi^2_0) \nonumber\\=
&\mathbf{E}\left[\sum_{t=0}^{T-1}\left(c_t\left(X^{1,g}_t\right)+c_t\left(X^{2,g}_t\right)\right)\right]\nonumber\\
=& \mathbf{E}\left[\sum_{t=0}^{T-2}\left(c_{t+1}\left(X^{1,g}_{t+1}\right)+c_{t+1}\left(X^{2,g}_{t+1}\right)\right)\right]\nonumber\\
&+\mathbf{E}\left[c_{0}\left(X^1_{0}\right)+c_{0}\left(X^2_{0}\right)\right] \nonumber\\
=& \mathbf{E}\left[\sum_{t=0}^{T-2}\rho_t\left(\overline{X}^{1,g}_{t},\overline{X}^{2,g}_{t},U^1_{t},U^2_{t}\right) \right]+\mathbf{E}\left[c_{0}\left(X^1_{0}\right)+c_{0}\left(X^2_{0}\right)\right] .
\end{align}
With the above notation and definition of system state and instantaneous holding cost, we have a dynamic team problem with non-classical information structure where the common information between the two controllers at any time $t$ is their decisions/control actions up to time $t-1$. This information structure is the control sharing information structure investigated in \cite{mahajan2012optimal}. 
Furthermore, the independent assumption we made about the exogenous arrivals and the service processes is the same as the assumptions made about the noise variables in \cite{mahajan2012optimal}. 
Therefore, the following Properties \ref{P:indep}-\ref{P:sep2} hold by the results in \cite{mahajan2012optimal}.
\begin{prop} 
\label{P:indep}
For each $t$, and any given $g^1_s(.),g^2_s(.),s\leq t$, we have 
\begin{align}
&\mathbf{P}\left(I^1_t=i^1_t,I^2_t=i^2_t | U^{1}_{0:t-1},U^{2}_{0:t-1}\right) \nonumber\\
=& \mathbf{P}\left(I^1_t=i^1_t| U^{1}_{0:t-1},U^{2}_{0:t-1}\right)  \mathbf{P}\left(I^2_t=i^2_t| U^{1}_{0:t-1},U^{2}_{0:t-1}\right).
\end{align}
\end{prop}
\begin{proof}
Same as that of Proposition 2 in \cite{mahajan2012optimal}.
\end{proof}
Property \ref{P:indep} says that the two subsystems are independent conditional on past control actions. 

Because of Property \ref{P:indep} and \eqref{SR:eqdynamic}, each controller $C_i, i=1,2$ can generate its decision at any time $t$ by using only its current local state $\overline{X}^i_t$ and past decisions of both controllers. This assertion is established by the following property.
\begin{prop} 
\label{P:sep1}
For the routing problems formulated in Section \ref{sec:model}, without loss of optimality we can restrict attention to routing policies of the form
\begin{align}
U^1_t =g^1_t\left(\overline{X}^1_t,U^{1}_{0:t-1},U^{2}_{0:t-1}\right),\\
U^2_t =g^2_t\left(\overline{X}^2_t,U^{1}_{0:t-1},U^{2}_{0:t-1}\right).
\end{align}
\end{prop}
\begin{proof}
Same as that of Proposition 1 in \cite{mahajan2012optimal}.
\end{proof}
Using the common information approach in \cite{nayyar2013decentralized},
we can refine the result of Property \ref{P:sep1} as follows.
\begin{prop}
\label{P:sep2}
For the two routing problems formulated in Section \ref{sec:model}, without loss of optimality we can restrict attention to routing policies of the form
\begin{align}
U^1_t =g^1_t\left(\overline{X}^1_t,\overline{\Pi}^1_t,\overline{\Pi}^2_t\right),\label{P:sep2:equ1}\\
U^2_t =g^2_t\left(\overline{X}^1_t,\overline{\Pi}^1_t,\overline{\Pi}^2_t\right).\label{P:sep2:equ2}
\end{align}
\end{prop}
\begin{proof}
Same as that of Theorem 1 in \cite{mahajan2012optimal}.
\end{proof}
The result of Property \ref{P:sep2} will play a central role in the analysis of the decentralized routing problems formulated in this paper.

\section{The Decentralized Policy $\hat{g}$ and Preliminary results}
\label{sec:prelim}
In this section, we specify a decentralized policy $\hat{g}$ and 
identify an information state for each controller.
Furthermore, we develop some preliminary results for both the finite horizon problem and the infinite horizon problem formulated in Section \ref{sec:model}.

To specify policy $\hat{g}$, we first define the upper bound and lower bound on the support of the PMF, $\Pi^i_t, i=1,2$ as
\begin{align}
& UB^i_t := \max(x:\Pi^i_t(x) \neq 0),\label{eq:UBi}\\
& LB^i_t := \min(x:\Pi^i_t(x) \neq 0).\label{eq:LBi}\\
& UB_t := \max(UB^1_t,UB^2_t),\label{eq:UB}\\
& LB_t := \min(LB^1_t,LB^2_t).\label{eq:LB}
\end{align}
Similarly, we define the bounds on the support of the PMF, $\overline{\Pi}^i_t, i=1,2$ as
\begin{align}
& \overline{UB}^i_t := \max(x:\overline{\Pi}^i_t(x) \neq 0),\label{eq:UBbari}\\
& \overline{LB}^i_t := \min(x:\overline{\Pi}^i_t(x) \neq 0),\label{eq:LBbari}\\
& \overline{UB}_t := \max(\overline{UB}^1_t,\overline{UB}^2_t),\label{eq:UBbar}\\
& \overline{LB}_t := \min(\overline{UB}^1_t,\overline{UB}^2_t).\label{eq:LBbar}
\end{align}

Using the above bounds, we specify the policy $\hat{g}:= (\hat{g}^1,\hat{g}^2)$ as follows: 
\begin{align}
&U^i_t = \hat{g}^i_t\left(\overline{X}^i_t,\overline{UB}_t,\overline{LB}_t\right) 
= \left\lbrace\begin{array}{l}
1, \text{ when } \overline{X}^i_t \geq  TH_t,\\
0, \text{ when } \overline{X}^i_t < TH_t,
\end{array}\right.
\label{eq:ghat}
\end{align}
where
\begin{align}
& TH_t=
\frac{1}{2}\left(\overline{UB}_t+\overline{LB}_t\right). \label{eq:th}
\end{align}
Under $\hat{g}$, each controller routes a customer to the other queue when
$\overline{X}^i_t,i=1,2$, the queue length of its own station at the time of decision, is greater than or equal to the threshold given by (\ref{eq:th}).

Note that this decentralized routing policy $\hat{g}$ is indeed of the form asserted by Property \ref{P:sep2}
since the upper and lower bounds $\overline{UB}_t$ and $\overline{LB}_t$ are both functions of 
the PMFs $\overline{\Pi}^1_t,\overline{\Pi}^2_t$.
Therefore, the threshold $TH_t$, as a function of $\overline{\Pi}^1_t,\overline{\Pi}^2_t$, is common knowledge between the controllers. 
Using the common information,
each controller can compute the threshold according to 
(\ref{eq:th}) individually, and $\hat{g}$ can be implemented in a decentralized manner.

Under policy $\hat{g}$, the evolution of the bounds defined by (\ref{eq:UB})-(\ref{eq:LBbar}) are determined by the following lemma.
\begin{lemma}
\label{lm:bounds}
At any time $t$ we have
\begin{align}
&\overline{UB}^{\hat{g}}_{t} = UB^{\hat{g}}_{t}+1 , \quad
\overline{LB}^{\hat{g}}_{t} = \left(LB^{\hat{g}}_{t}-1 \right)^+ .
\label{eq:UBLBbarupdate}
\end{align}
When $(U^{1,\hat{g}}_t,U^{2,\hat{g}}_t)=(0,0)$
\begin{align}
&UB^{\hat{g}}_{t+1} = \left \lceil TH_t \right \rceil -1 , \quad
LB^{\hat{g}}_{t+1} = \overline{LB}^{\hat{g}}_{t}
\label{eq:UBLBupdate00}
\end{align}
When $(U^{1,\hat{g}}_t,U^{2,\hat{g}}_t)=(1,1)$
\begin{align}
&UB^{\hat{g}}_{t+1} = \overline{UB}^{\hat{g}}_{t} , \quad
LB^{\hat{g}}_{t+1} = \left \lceil TH_t \right \rceil
\label{eq:UBLBupdate11}
\end{align}
When $(U^{i,\hat{g}}_t,U^{j,\hat{g}}_t)=(1,0), i =1,2, j\neq i$
\begin{align}
&UB^{\hat{g}}_{t+1} = \max\left(\overline{UB}^{i,\hat{g}}_{t} -1, 
\left \lceil TH_t \right \rceil \right)  \\
&LB^{\hat{g}}_{t+1} = \min\left( \overline{LB}^{j,\hat{g}}_{t}+1,
\left\lceil TH_t \right \rceil-1 \right)
\label{eq:UBLBupdate10}
\end{align}
where $\lfloor x \rfloor = \text{maximum integer}\leq x$, and $\lceil x \rceil = \text{minimum integer}\geq x$.
\begin{flushright}
$\square$
\end{flushright}
\end{lemma}
\begin{proof}
See Appendix \ref{app:prelim}
\end{proof}

Corollary \ref{cor:difdec} below follows directly form (\ref{eq:UBLBbarupdate})-(\ref{eq:UBLBupdate10}) in Lemma \ref{lm:bounds}.
\begin{corollary}
\label{cor:difdec}
Under policy $\hat{g}$, 
\begin{align}
&UB^{\hat{g}}_{t+1}-LB^{\hat{g}}_{t+1} \nonumber\\
\leq & \left\lbrace
\begin{array}{ll}
\left\lceil\frac{1}{2}\left(UB^{\hat{g}}_t -LB^{\hat{g}}_t\right)\right \rceil & \text{ when }(U^{1,\hat{g}}_t,U^{2,\hat{g}}_t)=(0,0),\\
 UB^{\hat{g}}_{t}  -LB^{\hat{g}}_{t} &\text{ otherwise.}
\end{array}\right. 
\label{eq:UBLB}
\end{align}
Moreover, if $UB^{\hat{g}}_{t_0}-LB^{\hat{g}}_{t_0} \leq 1$ for some time $t_0$, then
\begin{align}
\left(UB^{\hat{g}}_t-LB^{\hat{g}}_t\right)  \leq 1 \text{ for all } t \geq t_0.
\label{eq:UBLBevenless1}
\end{align}
\begin{flushright}
$\square$
\end{flushright}
\end{corollary}
Corollary \ref{cor:difdec} shows that
the difference between the highest possible number of customers in $Q_1$ or $Q_2$ and the lowest possible number of customers in $Q_1$ or $Q_2$ is non-increasing under the policy $\hat{g}$. 
Furthermore, the difference is reduced by half when there is no customer routed from one queue to another one.



\section{The finite horizon problem}
\label{sec:finite}
In this section, we consider the finite horizon problem formulated in Section \ref{sub:model:finite},
under the additional condition $X^1_0 =X^2_0=x_0$, where $x_0$ is arbitrary but fixed, and is common knowledge between $C_1$ and $C_2$.
\subsection{Analysis}
\label{sub:finite:anal}

The main result of this section asserts that the policy $\hat{g}$ defined in Section \ref{sec:prelim} is optimal.
\begin{theorem}
\label{thm:opt}
When $X^1_0 =X^2_0=x_0$ and $x_0$ is common knowledge between $C_1$ and $C_2$, the policy $\hat{g}$ given by (\ref{eq:ghat})-(\ref{eq:th}) is optimal for the finite horizon decentralized routing problem formulated in Section \ref{sub:model:finite}, that is 
\begin{align}
J^{\hat{g}}_T(x_0,x_0) \leq J^{g}_T(x_0,x_0)
\end{align}
for any feasible policy $g \in \mathcal{G}_d $ and any initial queue length $x_0$.
\begin{flushright}
$\square$
\end{flushright}
\end{theorem}

Before proving Theorem \ref{thm:opt}, we note that when $X^1_0 =X^2_0=x_0$ Corollary \ref{cor:difdec} implies that
\begin{align}
UB^{\hat{g}}_t-LB^{\hat{g}}_t  \leq 1 \text{ for all }t \geq 0.
\label{eq:finite:dif1}
\end{align}
Equation (\ref{eq:finite:dif1}) says that the difference between the highest possible number of customers in $Q_1$ or $Q_2$ and the lowest possible number of customers in $Q_1$ or $Q_2$ is less than or equal to $1$ under policy $\hat{g}$. 
This property means that $\hat{g}$ controls the length of the joint support of the PMFs $\overline{\Pi}^1_t, \overline{\Pi}^2_t$ and balances the lengths of the two queues.
A direct consequence of (\ref{eq:finite:dif1}) is the following corollary.
\begin{corollary}
\label{cor:bal}
At any time $t$, we have
\begin{align}
&\left\lfloor \frac{1}{2}(X^{1,\hat{g}}_t+X^{2,\hat{g}}_t )\right\rfloor = \min\left(X^{1,\hat{g}}_t,X^{2,\hat{g}}_t\right),
\\
&\left\lceil \frac{1}{2}(X^{1,\hat{g}}_t+X^{2,\hat{g}}_t)\right\rceil = \max\left(X^{1,\hat{g}}_t,X^{2,\hat{g}}_t\right).
\end{align}
\begin{flushright}
$\square$
\end{flushright}
\end{corollary}
As pointed out above, the policy $\hat{g}$ balances the lengths of the two queues. 
This balancing property suggests that the throughput of the system due to $\hat{g}$ is high and the total number of customers in the system is low. This is established by the following lemma.
\begin{lemma}
\label{lm:stoless}
Under the assumption $X^1_0=X^2_0=x_0$, where $x_0$ is common knowledge, for any policy $g$ of the form described by 
(\ref{P:sep2:equ1})-(\ref{P:sep2:equ2}), we have
\begin{align}
X^{1,\hat{g}}_t+X^{2,\hat{g}}_t \leq_{st} X^{1,g}_t+X^{2,g}_t,
\end{align}
where $Z_1\leq_{st}Z_2$ means that the r.v. $Z_1$ is stochastically smaller than 
the r.v. $Z_2$, that is, for any $a\in \mathbb{R}$, $\mathbf{P}(Z_1\geq a) \leq \mathbf{P}(Z_2 \geq a)$ (see \cite{marshall2010inequalities}).
\begin{flushright}
$\square$
\end{flushright}
\end{lemma}
\begin{proof}
See Appendix \ref{app:finite}
\end{proof}
Using Lemma \ref{lm:stoless}, we now prove Theorem \ref{thm:opt}.
\begin{proof}[Proof of Theorem \ref{thm:opt}]
For any feasible policy $g$, since the functions $c_t,t=0,1,...,T$, are convex, we have at any time $t$
\begin{align}
&\mathbf{E}\left[c_t\left(X^{1,g}_t\right)+c_t\left(X^{2,g}_t\right)\right]\nonumber\\
\geq & \mathbf{E}\left[c_t\left(\left\lfloor \frac{1}{2}(X^{1,g}_t+X^{2,g}_t)\right\rfloor\right)+c_t\left(\left\lceil \frac{1}{2}(X^{1,g}_t+X^{2,g}_t)\right\rceil\right)\right].
\label{pfthmopt:eq1}
\end{align}
Furthermore, using Lemma \ref{lm:stoless} and the fact that $c_t(\cdot)$ is increasing, we get
\begin{align}
&\mathbf{E}\left[c_t\left(\left\lfloor \frac{1}{2}(X^{1,g}_t+X^{2,g}_t)\right\rfloor\right)+c_t\left(\left\lceil \frac{1}{2}(X^{1,g}_t+X^{2,g}_t)\right\rceil\right)\right] \nonumber\\
\geq & \mathbf{E}\left[c_t\left(\left\lfloor \frac{1}{2}(X^{1,\hat{g}}_t+X^{2,\hat{g}}_t)\right\rfloor\right)
+c_t\left(\left\lceil \frac{1}{2}(X^{1,\hat{g}}_t+X^{2,\hat{g}}_t)\right\rceil\right)\right] \nonumber\\
= & \mathbf{E}\left[c_t\left(\min(X^{1,\hat{g}}_t,X^{2,\hat{g}}_t)\right)+c_t\left(\max (X^{1,\hat{g}}_t,X^{2,\hat{g}}_t)\right)\right] \nonumber\\
= & \mathbf{E}\left[c_t\left(X^{1,\hat{g}}_t\right)+c_t\left(X^{2,\hat{g}}_t\right)\right]. \label{pfthmopt:eq2}
\end{align} 
The inequality in (\ref{pfthmopt:eq2}) is true because 
 $X^{1,g}_t+X^{2,g}_t\leq_{st} X^{1,\hat{g}}_t+X^{2,\hat{g}}_t$ (Lemma \ref{lm:stoless}) and $c_t(\cdot)$ is increasing.
The first equality in (\ref{pfthmopt:eq2}) follows from Corollary \ref{cor:bal}.\\
Combining (\ref{pfthmopt:eq1}) and (\ref{pfthmopt:eq2}) we obtain, for any $t$, 
\begin{align}
&\mathbf{E}\left[c_t\left(X^{1,g}_t\right)+c_t\left(X^{2,g}_t\right)\right] \geq  \mathbf{E}\left[c_t\left(X^{1,\hat{g}}_t\right)+c_t\left(X^{2,\hat{g}}_t\right)\right].
\label{pfthm:instopt}
\end{align}
The optimality of policy $\hat{g}$ follows from (\ref{Model:finite:cost}) and (\ref{pfthm:instopt}).
\end{proof}

\subsection{Comparison to the performance under centralized information}
\label{sub:finite:comp}
We compare now the performance of the optimal decentralized policy $\hat{g}$ to the performance of the queueing system under centralized information.
The results of this comparison will be useful when we study the infinite horizon problem in Section \ref{sec:infinite}.

Consider a centralized controller who has all the information $I^1_t$ and $I^2_t$ at each time $t$.
Then, the set $\mathcal{G}_c$ of feasible routing policies of the centralized controller is 
\begin{align}
\mathcal{G}_c :=
\{(g^1,g^2):\quad & g^i = (g^i_0,g^i_1,\dots,g^i_t,\dots),i=1,2 \nonumber\\
		&\text{and } U^i_t=g^i_t(I^1_t,I^2_t)\}.
\label{Model:cenpolicies}
\end{align}
By the definition, $\mathcal{G}_d \subset \mathcal{G}_c$.
This means that the centralized controller can simulate any decentralized policy $g \in \mathcal{G}_d$ adopted by controllers $C_1$ and $C_2$.
Therefore, for any initial PMFs $\pi^1_0,\pi^2_0$
\begin{align}
&\inf_{g\in \mathcal{G}_c} J^g_T(\pi^1_0,\pi^2_0) \leq \inf_{g\in \mathcal{G}_d} J^g_T(\pi^1_0,\pi^2_0) 
\label{comp:eqJ1}
\\
&\inf_{g\in \mathcal{G}_c} J^g(\pi^1_0,\pi^2_0) \leq \inf_{g\in \mathcal{G}_d} J^g(\pi^1_0,\pi^2_0) .
\label{comp:eqJinf}
\end{align}
When $X^1_0=X^2_0=x_0$, Lemma \ref{lm:stoless} and Theorem \ref{thm:opt} 
show that the cost given by $\hat{g}$ is smaller than the cost given by any policy $g \in \mathcal{G}_d$.
Furthermore we have:
\begin{lemma}
\label{lm:comp}
Under the assumption $X^1_0=X^2_0=x_0$, where $x_0$ is common knowledge, we have
\begin{align}
X^{1,\hat{g}}_t+X^{2,\hat{g}}_t \leq_{st} X^{1,g}_t+X^{2,g}_t,
 \label{comp:eqsto}
\end{align}
for any $g \in \mathcal{G}_c$, and 
\begin{align}
 J^{\hat{g}}_T(x_0,x_0)  \leq \inf_{g\in \mathcal{G}_c} J^g_T(x_0,x_0).
 \label{comp:eqJ2}
\end{align}
for any $g \in \mathcal{G}_c$.
\begin{flushright}
$\square$
\end{flushright}
\end{lemma}
\begin{proof}
The proof of \eqref{comp:eqsto} is the same as the proof of Lemma \ref{lm:stoless}, and the proof of \eqref{comp:eqJ2} is the same as the proof of Theorem \ref{thm:opt}.
\end{proof}

Since $\hat{g}$ is a decentralized policy, \eqref{comp:eqJ1} and Lemma \ref{lm:comp} imply that
\begin{align}
 J^{\hat{g}}_T(x_0,x_0)  
 = \inf_{g\in \mathcal{G}_d} J^g_T(x_0,x_0)=\inf_{g\in \mathcal{G}_c} J^g_T(x_0,x_0).
\label{comp:eqJeq}
\end{align}
Equation (\ref{comp:eqJeq}) shows that when $X^1_0=X^2_0=x_0$ and $x_0$ is common knowledge between $C_1$and $C_2$, policy $\hat{g}$ achieves the same performance as any centralized optimal policy.
\subsection{The Case of Different Initial Queue Lengths}\label{sub:finite:diss}
%

When $X^1_0\neq X^2_0$, the policy $\hat{g}$ is not necessarily optimal for the finite horizon problem.

Consider an example where the horizon $T=1$ (two-step horizon), $\lambda = 0.1, \mu =0.5$ and
\begin{align}
&\mathbf{P}\left(X^1_0 = 3\right) = 1,\\
&\mathbf{P}\left(X^2_0=1\right)=0.9, \quad \mathbf{P}\left(X^2_0=5\right)=0.1, \quad
\end{align}
that is,
\begin{align}
\pi^1_0 =& (0,0,0,1,0,0,0,\dots),\\
\pi^2_0 =& (0,0.9,0,0,0,0.1,0,\dots),
\end{align}
where $\pi^1_0,\pi^2_0$ denote the initial PMFs on the lengths of the queues. \\
Then, $\overline{\Pi}^1_0,\overline{\Pi}^2_0$ and the threshold $TH_0$ are
\begin{align}
&\overline{\Pi}^1_0 = (0,0,0.5,0.4,0.1,0,0,\dots), \label{remark:ex:eqpi1}\\
&\overline{\Pi}^2_0 = (0.45,0.36,0.09,0,0.05,0.04,0.01,\dots),\label{remark:ex:eqpi2}\\
&TH_0= \frac{1}{2}(6+0)=3.\label{remark:ex:eqth}
\end{align}
Consider the cost functions $c_0(x)=0$ and $c_1(x)=x^2$. Then, we have
\begin{align}
&J^g(\pi^1_0,\pi^2_0) \nonumber\\
= & \mathbf{E}\left[\left(X^{1,g}_1\right)^2+\left(X^{2,g}_1\right)^2\right] \nonumber\\
    = & \mathbf{E}\left[\left(\overline{X}^1_0-U^{1,g}_0+U^{2,g}_0\right)^2
    	+\left(\overline{X}^2_0-U^{2,g}_0+U^{1,g}_0\right)^2\right].
\end{align}
Using (\ref{remark:ex:eqpi1})-(\ref{remark:ex:eqth}) and the specification of the policy $\hat{g}$,
we can compute the expected cost due to $\hat{g}$. It is
\begin{align}
J^{\hat{g}}(\pi^1_0,\pi^2_0)  = & 8.48.
\end{align}
Consider now another policy $\tilde{g}$ described below. For $i=1,2,i\neq j$,
\begin{align}
&U^{i,\tilde{g}}_t = \tilde{g}_t\left(\overline{X}^i_t,\overline{\Pi}^1_t,\overline{\Pi}^2_t\right) 
= \left\lbrace\begin{array}{l}
1, \text{ when } \overline{X}^i_t \geq  \mathbf{E}\left[\overline{X}^j_t|\overline{\Pi}^j_t\right],\\
0, \text{ when } \overline{X}^i_t < \mathbf{E}\left[\overline{X}^j_t|\overline{\Pi}^j_t\right],
\end{array}\right.
\label{eq:gt}
\end{align}
Then, from \eqref{remark:ex:eqpi1}-\eqref{remark:ex:eqpi2} and \eqref{eq:gt} we get
\begin{align}
&U^{1,\tilde{g}}_0 = \left\lbrace\begin{array}{l}
1, \text{ when } \overline{X}^1_0 \geq  1,\\
0, \text{ when } \overline{X}^1_0 < 1,
\end{array}\right.\\
&U^{2,\tilde{g}}_0 = \left\lbrace\begin{array}{l}
1, \text{ when } \overline{X}^2_0 \geq  2.6,\\
0, \text{ when } \overline{X}^2_0 < 2.6,
\end{array}\right.
\end{align}
Therefore, the expected cost due to the policy $\tilde{g}$ is given by
\begin{align}
J^{\tilde{g}}(\pi^1_0,\pi^2_0)  = & 8.28
\end{align}
Since $J^{\tilde{g}}(\pi^1_0,\pi^2_0)=8.28<8.48=J^{\hat{g}}(\pi^1_0,\pi^2_0)$, policy $\hat{g}$ is not optimal.

In this example, each controller has only one decision to make, the decision at time $0$.
As a result, signaling does not provide any advantages to the controllers, and that is why the policy $\hat{g}$ is not the best policy.

\section{Infinite horizon}
\label{sec:infinite}
We consider the infinite horizon decentralized routing problem formulated in Section \ref{sub:model:infinite}, and make the following additional assumptions.
\begin{assumption}
\label{assum:mulambda}
$\mu>\lambda$.
\end{assumption}
\begin{assumption}
\label{assum:initial}
The initial PMFs
$\pi^1_0, \pi^2_0$ are finitely supported and common knowledge between controllers $C_1$ and $C_2$. 
i.e. there exists $M<\infty$ such that $\pi^1_0(x)=\pi^2_0(x)=0$ for all $x>M$.
\end{assumption}

Let $g_0$ denote the open-loop policy that does not do any routing, that is, at any time $t$
\begin{align}
U^{1,g_0}_t = U^{2,g_0}_t = 0.
\end{align}

\begin{assumption}
\label{assum:cost}
\begin{align}
&\lim_{T \rightarrow \infty}
\frac{1}{T} J^{g_0}_T(\pi^1_0,\pi^2_0) :=J^{g_0} < \infty  \quad a.s.,
\end{align}
where $J^{g_0}$ is a constant that denotes the infinite horizon average cost per unit time due to policy $g_0$.
\end{assumption}

\begin{remark}
Due to policy $g_0$, the queue length $\{X^{g_0,i}_t, t\in \mathbb{Z}_+ \}, i=1,2$ is a positive recurrent birth and death chain with arrival rate $\lambda$ and departure rate $\mu 1_{\{X^{g_0,i}_t \neq 0 \}}$. Therefore, as $T\rightarrow \infty$, the average cost per unit time converges to a constant a.s. if the expected cost under the stationary distribution of the process is finite (see \cite[chap. 3]{bremaud1999markov}).
Assumption \ref{assum:cost} is equivalent to the assumption that the expected cost is finite under the stationary distribution of the controlled queue lengths.
\end{remark}

We proceed to analyze the infinite horizon average cost per unit time for the model of Section \ref{sec:model} under Assumptions \ref{assum:mulambda}-\ref{assum:cost}.

\subsection{Analysis} \label{sub:infinite:anal}
When $X^1_0\neq X^2_0$, the policy $\hat{g}$, defined in Section \ref{sec:prelim}, is not necessarily optimal for the finite horizon problem (see the example in Section \ref{sub:finite:diss}).
Nevertheless, the policy $\hat{g}$ still attempts to balance the queues. 
Given enough time, policy $\hat{g}$ may be able to balance the queue lengths even if they are not initially balanced.
In this section we show that this is indeed the case. 

Specifically, we prove the optimality of policy $\hat{g}$ for the infinite horizon average cost per unit time problem, as stated in the following theorem which is the main result of this section.

\begin{theorem}
\label{thm:inf}
Under Assumptions \ref{assum:mulambda}-\ref{assum:cost}, the policy $\hat{g}$, described by (\ref{eq:ghat})-(\ref{eq:th}), is optimal for the infinite horizon average cost per unit time problem formulated in Section \ref{sub:model:infinite}.
\begin{flushright}
$\square$
\end{flushright}
\end{theorem}

To establish the assertion of Theorem \ref{thm:inf} we proceed in four steps. 
In the first step we show that the infinite horizon average cost per unit time due to policy $\hat{g}$ is bounded above by the cost of the uncontrolled queues (i.e. the cost due to policy $g_0$). 
In the second step we show that under policy $\hat{g}$ the queues are eventually balanced, i.e. the queue lengths can differ by at most one.
In the third step we derive a result that connects the performance of policy $\hat{g}$ under the initial PMFs $(0,0)$ to the performance of the optimal policy under any arbitrary initial PMFs $\pi^1_0,\pi^2_0$ on queues $Q_1$ and $Q_2$.
In the forth step we establish the optimality of policy $\hat{g}$ based on the results of steps one, two and three.

\subsubsection*{Step 1}
We prove that $J^{\hat{g}}(\pi^1_0,\pi^2_0) \leq J^{g_0}$. To do this, we first establish some preliminary results that appear in 
Lemmas \ref{lm:inf:stobdd} and \ref{lm:inf:bddunctrl}.
\begin{lemma}
\label{lm:inf:stobdd}
There exists processes $\{Y^{1}_t, t\in \mathbb{Z}_+\}$ and $\{Y^{2}_t, t\in \mathbb{Z}_+\}$ such that
\begin{align}
&\{Y^{i}_t, t\in \mathbb{Z}_+\}
\text{ has the same distribution as } 
\{X^{i,g_0}_t, t\in \mathbb{Z}_+\}
\label{eq:stoeq}
\end{align}
for $i=1,2$, and for all times $t$
\begin{align}
& X^{1,\hat{g}}_t+X^{2,\hat{g}}_t \leq Y^{1}_t+Y^{2}_t \quad a.s. , \label{eq:stosumbdd}\\
& \max_i\left(X^{i,\hat{g}}_t\right)  \leq \max_i\left(Y^{i}_t\right)\quad a.s. \label{eq:stomaxbdd}
\end{align}
\begin{flushright}
$\square$
\end{flushright}
\end{lemma}
\begin{proof}
See Appendix \ref{app:infinite1}
\end{proof}

Lemma \ref{lm:inf:stobdd} means that the uncontrolled queue lengths are longer than the queue lengths under policy $\hat{g}$ in a stochastic sense. Note that (\ref{eq:stosumbdd}) and (\ref{eq:stomaxbdd}) are not true if $Y^i_t, i=1,2$, is replaced by $X^{i,g_0}_t, i=1,2$,
as the following example shows.
\\
\underline{Example}
\\
When $X^{1,g_0}_t=4,X^{2,g_0}_t=6$ and $X^{1,\hat{g}}_t=X^{2,\hat{g}}_t=5$, the analogues of (\ref{eq:stosumbdd}) and (\ref{eq:stomaxbdd}) 
where $Y^i_t$ are replaced by $X^{i,g_0}_t, i=1,2$ are
\begin{align}
& X^{1,\hat{g}}_t+X^{2,\hat{g}}_t = X^{1,g_0}_t+X^{2,g_0}_t = 10 ,\\
& \max_i\left(X^{i,\hat{g}}_t\right) =5  \leq 6 = \max_i\left(X^{i,g_0}_t\right).
\end{align}
However, if $A^1_{t+1}=1,A^2_{t+1}=0$ and $D^1_{t+1}=0,D^2_{t+1}=1$ we get
$X^{1,g_0}_{t+1}=X^{2,g_0}_{t+1}=5$ and $X^{1,\hat{g}}_{t+1}=6, X^{2,\hat{g}}_{t+1}=4$, then
\begin{align}
& \max_i\left(X^{i,\hat{g}}_{t+1}\right) =6 > 5=  \max_i\left(X^{i,g_0}_{t+1}\right),
\end{align}
and the analogue of (\ref{eq:stomaxbdd}), when $Y^i_t$ is replaced by $X^{i,g_0}_t, i=1,2$, does not hold.

The stochastic dominance relation asserted by Lemma \ref{lm:inf:stobdd} implies that the instantaneous cost under policy $\hat{g}$ is almost surely no greater than the instantaneous cost due to policy $g_0$. This implication is made precise by the following lemma.

\begin{lemma}
\label{lm:inf:bddunctrl}
The processes $\{Y^{1}_t, t\in \mathbb{Z}_+\}$ and $\{Y^{2}_t, t\in \mathbb{Z}_+\}$ defined in Lemma \ref{lm:inf:stobdd}
are such that at any time $t$ 
\begin{align}
c\left(X^{1,\hat{g}}_t\right)+c\left(X^{2,\hat{g}}_t\right) \leq 
c\left(Y^{1}_t\right)+c\left(Y^{2}_t\right) \quad a.s.
\end{align}
\begin{flushright}
$\square$
\end{flushright}
\end{lemma}
\begin{proof}
See Appendix \ref{app:infinite1}
\end{proof}
In order to apply the result of Step 1 as the time horizon goes to infinity, we need the following result on the convergence of the cost due to $\{Y^{1}_t, t\in \mathbb{Z}_+\}$ and $\{Y^{2}_t, t\in \mathbb{Z}_+\}$.
\begin{lemma}
\label{lm:inf:WTconv}
Let $\{Y^{1}_t, t\in \mathbb{Z}_+\}$ and $\{Y^{2}_t, t\in \mathbb{Z}_+\}$ be the processes defined in Lemma \ref{lm:inf:stobdd}.
Let $W_T$ denote 
\begin{align}
W_T := \frac{1}{T}\sum_{t=0}^{T-1}\left(c(Y^1_t)+c(Y^2_t)\right).
\end{align}
Under Assumptions \ref{assum:initial} and \ref{assum:cost},
\begin{align}
\lim_{T \rightarrow \infty} W_T = J^{g_0} \quad a.s.
\end{align}
Moreover, $\{W_T, T=1,2,\dots\}$ is uniformly integrable, so it also converges in expectation.
\begin{flushright}
$\square$
\end{flushright}
\end{lemma}
\begin{proof}
See Appendix \ref{app:infinite1}
\end{proof}
A direct consequence of Lemmas \ref{lm:inf:stobdd}, \ref{lm:inf:bddunctrl} and \ref{lm:inf:WTconv} is the following.

\begin{corollary}
\label{cor:infcostfinite}
If $\lim_{T \rightarrow \infty} \frac{1}{T} \sum_{t=0}^{T-1}\left(c\left(X^{1,\hat{g}}\right)+c\left(X^{2,\hat{g}}\right)\right)$ converges a.s., then,
\begin{align}
\frac{1}{T} \sum_{t=0}^{T-1}\left(c\left(X^{1,\hat{g}}\right)+c\left(X^{2,\hat{g}}\right)\right)
\longrightarrow J^{\hat{g}}(\pi^1_0,\pi^2_0) 
\end{align}
in expectation and a.s. as $T\rightarrow \infty$. Furthermore,
\begin{align}
J^{\hat{g}}(\pi^1_0,\pi^2_0) \leq J^{g_0}<\infty.
\end{align}
\begin{flushright}
$\square$
\end{flushright}
\end{corollary}
\begin{proof}
See Appendix \ref{app:infinite1}
\end{proof}

\subsubsection*{Step 2}

We prove that under policy $\hat{g}$ the queues are eventually balanced. For this matter we first establish some preliminary results that appear in 
Lemmas \ref{lm:inf:MC} and \ref{lm:inf:00io}.
\begin{lemma}
\label{lm:inf:MC}
Let $T_0$ be a stopping time with respect to the process $\{X^{1,\hat{g}}_t,X^{2,\hat{g}}_t,t\in \mathbb{Z}_+\}$.
Define the process $\{S_t=S^{\hat{g}}_t,t\geq T_0+1\}$ as follows.
\begin{align}
S_{T_0+1} :=& X^{1,\hat{g}}_{T_0+1}+X^{2,\hat{g}}_{T_0+1} \label{pflm:inf:MC:eqS1}\\
S_{t+1} :=& S_t-D^1_t -D^2_t +A^1_t+A^2_t \nonumber\\
         &+ 1_{\{S_t=1\}}\left(1_{\left\{X^{1,\hat{g}}_t=0\right\}}(D^1_t- D^2_t)+D^2_t \right)\nonumber\\
         &+ 1_{\{S_t=0\}}\left(D^1_t+D^2_t \right)\label{pflm:inf:MC:eqS2}
\end{align}
If $\mu> \lambda>0$, then $\{S_t,t\geq T_0+1\}$ is an irreducible positive recurrent Markov chain. 
\begin{flushright}
$\square$
\end{flushright}
\end{lemma}
\begin{proof}
See Appendix \ref{app:infinite2}
\end{proof}
Lemma \ref{lm:inf:MC} holds for arbitrary stopping time $T_0$ with respect to $\{X^{1,\hat{g}}_t,X^{2,\hat{g}}_t,t\in \mathbb{Z}_+\}$. 
By appropriately selecting $T_0$ we will show later that $S_t$ is coupled with $X^{1,\hat{g}}_t+X^{2,\hat{g}}_t$, i.e.
for all $t>T_0$, $S_t=X^{1,\hat{g}}_t+X^{2,\hat{g}}_t$.
This result along with the fact that the process $\{S_t,t\geq T_0+1\}$ is an irreducible positive recurrent Markov chain will allow us to analyze the cost due to policy $\hat{g}$.

\begin{lemma}

\label{lm:inf:00io}
Under policy $\hat{g}$, 
\begin{align}
\mathbf{P}\left(\left(U^{1,\hat{g}}_t,U^{2,\hat{g}}_t\right)=(0,0) \quad i.o.\right)  = 1.
\end{align}
\begin{flushright}
$\square$
\end{flushright}
\end{lemma}
\begin{proof}
See Appendix \ref{app:infinite2}
\end{proof}
Lemma \ref{lm:inf:00io} means that the event $\{ $ there exists $t_0<\infty$ such that at least one of the queue lengths is above the threshold defined by \eqref{eq:th} for all $t>t_0$  $\}$ can not happen. 
The idea of Lemma \ref{lm:inf:00io} is the following. 
If one of the queues, say $Q_1$, has length above the threshold, hence above the lower bound $LB^{\hat{g}}_t$, 
then, the length of $Q_2$ does not decrease, because under policy $\hat{g}$ $Q_2$ receives one customer from $Q_1$
and has at most one departure at this time.
Therefore, both queue lengths at the next time are bounded below by the current lower bound $LB^{\hat{g}}_t$.
When at least one of the queue lengths is above the threshold for all $t>t_0$, the queue lengths are bounded below by 
$LB^{\hat{g}}_{t_0}$ for all $t>t_0$. This kind of lower bound can not exist if the total arrival rate $2\lambda$ to the system is less than the total departure rate $2\mu$ from the system.

Lemma \ref{lm:inf:00io} and Corollary \ref{cor:difdec} in Section \ref{sec:prelim}
can be used to establish that under policy $\hat{g}$ the queues are eventually balanced.
This is shown in the corollary below.

\begin{corollary}
\label{cor:inf:dif1}
Let
\begin{align}
T_0 := \inf\{t : UB^{\hat{g}}_{t}-LB^{\hat{g}}_{t} \leq 1 \}.
\end{align}
Then
\begin{align}
\mathbf{P}(T_0<\infty) =1
\end{align}
and
\begin{align}
\left(UB^{\hat{g}}_t-LB^{\hat{g}}_t\right)  \leq 1 \text{ for all } t \geq T_0.
\label{eq:UBLBevenless1}
\end{align}

\begin{flushright}
$\square$
\end{flushright}
\end{corollary}

\subsubsection*{Step 3}
We compare the finite horizon cost $J^{\hat{g}}_T(0,0)$ (respectively, the infinite horizon cost $J^{\hat{g}}(0,0)$) due to policy $\hat{g}$ under initial PMFs $(0,0)$
to the minimum finite horizon cost $\inf_{g\in \mathcal{G}_d} J^g_T (\pi^1_0,\pi^2_0)$ (respectively, the minimum infinite horizon cost $\inf_{g\in \mathcal{G}_d} J^g (\pi^1_0,\pi^2_0)$) under arbitrary initial PMFs $(\pi^1_0,\pi^2_0)$.
\begin{lemma}
\label{lm:cen}
For any finite time $T$ and any initial PMFs $\pi^1_0,\pi^2_0$.
\begin{align}
J^{\hat{g}}_T(0,0)=
\inf_{g\in \mathcal{G}_c} J^g_T(0,0) \leq \inf_{g\in \mathcal{G}_c} J^g_T (\pi^1_0,\pi^2_0) \leq \inf_{g\in \mathcal{G}_d} J^g_T (\pi^1_0,\pi^2_0),
\end{align}
and
\begin{align}
J^{\hat{g}}(0,0)=
\inf_{g\in \mathcal{G}_c} J^g(0,0) \leq \inf_{g\in \mathcal{G}_c} J^g (\pi^1_0,\pi^2_0) \leq \inf_{g\in \mathcal{G}_d} J^g (\pi^1_0,\pi^2_0).
\end{align}
\begin{flushright}
$\square$
\end{flushright}
\end{lemma}
\begin{proof}
See Appendix \ref{app:lmcen}.
\end{proof}

Lemma \ref{lm:cen} states that the minimum cost achieved when the queues are initially empty is smaller than the minimum cost obtained when the system's initial condition is given by arbitrary PMFs on the lengths of queues $Q_1$ and $Q_2$.
This result is established through the use of the corresponding centralized information system that is discussed in Section \ref{sub:finite:comp}.

\subsubsection*{Step 4}
Based on the results of Steps 1, 2 and 3 we now establish the optimality of policy $\hat{g}$ for the infinite horizon average cost per unit time problem formulated in Section \ref{sub:model:infinite}.
First, we outline the key ideas in the proof of Theorem \ref{thm:inf}, then we present its proof.
Step 2 ensures that policy $\hat{g}$ eventually (in finite time) balances the queues.
Step 1 ensures that the cost $J^{\hat{g}}(\pi^1_0,\pi^2_0)$ is finite.
These two results together imply that the cost due to policy $\hat{g}$ is the same as the cost incurred after the queues are balanced.
Furthermore, we show that the cost of policy $\hat{g}$ is independent of the initial PMFs on the queue lengths.
Then, the result of Step 3 together with the results on the finite horizon problem establish the optimality of policy $\hat{g}$.


\begin{proof}[Proof of Theorem \ref{thm:inf}]
Define $T_0$ to be the first time when the length of the joint support of PMFs 
$\Pi^{1,\hat{g}}_t,\Pi^{2,\hat{g}}_t$ is no more than $1$.
That is
\begin{align}
T_0 = \inf\{t : UB^{\hat{g}}_{t}-LB^{\hat{g}}_{t} \leq 1 \}.
\end{align}
The random variable $T_0$ is a stopping time with respect to the process $\{X^{1,\hat{g}}_t,X^{2,\hat{g}}_t, t\in\mathbb{Z}_+\}$.
From Corollary \ref{cor:inf:dif1} we have
\begin{align}
&\mathbf{P}(T_0 < \infty )=1, \\
& UB^{\hat{g}}_{t}-LB^{\hat{g}}_{t} \leq 1 \text{ for all } t \geq T_0.
\end{align}
Furthermore, for all $t \geq T_0$
\begin{align}
&\left| X^{1,\hat{g}}_{t}-X^{2,\hat{g}}_{t}\right|
 \leq UB^{\hat{g}}_{t}-LB^{\hat{g}}_{t} \leq 1.
 \label{eq:pfinf:UBLBless1}
\end{align}
Consider the process $\{S_t,t\geq T_0+1\}$ defined by \eqref{pflm:inf:MC:eqS1} and \eqref{pflm:inf:MC:eqS2} (in Lemma \ref{lm:inf:MC}).
We claim that for all $t \geq T_0+1$
\begin{align}
X^{1,\hat{g}}_{t}+X^{2,\hat{g}}_{t}= S_t.
\label{eq:pfinf:claim}
\end{align}
We prove the claim in Appendix \ref{app:claiminfthminf}.
Suppose the claim is true.
Since $\left| X^{1,\hat{g}}_{t}-X^{2,\hat{g}}_{t}\right| \leq 1$ for all $t \geq T_0+1$,
the instantaneous cost at time $t \geq T_0+1$ is equal to 
\begin{align}
&c\left(X^{1,\hat{g}}_t\right)+c\left(X^{2,\hat{g}}_t\right)\nonumber\\
= & c\left(\left\lfloor \frac{1}{2}(X^{1,\hat{g}}_t+X^{2,\hat{g}}_t)\right\rfloor\right)
+c\left(\left\lceil \frac{1}{2}(X^{1,\hat{g}}_t+X^{2,\hat{g}}_t)\right\rceil\right)\nonumber\\
= & c\left(\left\lfloor \frac{1}{2}S^{\hat{g}}_t\right\rfloor\right)
+c\left(\left\lceil \frac{1}{2}S^{\hat{g}}_t\right\rceil\right).
\end{align}
Then, the average cost per unit time due to policy $\hat{g}$ is given by
\begin{align}
 &\frac{1}{T}\sum_{t=0}^{T-1}
\left(c\left(X^{1,\hat{g}}_t\right)+c\left(X^{2,\hat{g}}_t\right) \right)\nonumber\\
= &\frac{1}{T}
\sum_{t=0}^{T_0}
\left(c\left(X^{1,\hat{g}}_t\right)+c\left(X^{2,\hat{g}}_t\right)\right)\nonumber\\
&+\frac{1}{T}\sum_{t=T_0+1}^{T-1}
\left(c\left(\left\lfloor \frac{1}{2}S^{\hat{g}}_t\right\rfloor\right)
+c\left(\left\lceil \frac{1}{2}S^{\hat{g}}_t\right\rceil\right)\right).
\end{align}
Since $T_0<\infty \, a.s.$, we obtain
\begin{align}
&\lim_{T\rightarrow \infty}\frac{1}{T}\sum_{t=0}^{T-1}
\left(c\left(X^{1,\hat{g}}_t\right)+c\left(X^{2,\hat{g}}_t\right)\right) \nonumber\\
= &\lim_{T\rightarrow \infty}\frac{1}{T}\sum_{t=0}^{T_0}
\left(c\left(X^{1,\hat{g}}_t\right)
+c\left(X^{2,\hat{g}}_t\right) \right)\nonumber\\
&+\lim_{T\rightarrow \infty}\frac{1}{T}\sum_{t=T_0(+1}^{T-1}
\left(c\left(\left\lfloor \frac{1}{2}S^{\hat{g}}_t\right\rfloor\right)
+c\left(\left\lceil \frac{1}{2}S^{\hat{g}}_t\right\rceil\right) \right)\nonumber\\
= &\lim_{t\rightarrow \infty}\frac{1}{T}\sum_{t=T_0+1}^{T-1}
\left(c\left(\left\lfloor \frac{1}{2}S^{\hat{g}}_t\right\rfloor\right)
+c\left(\left\lceil \frac{1}{2}S^{\hat{g}}_t\right\rceil\right) \right)\nonumber\\
= & \sum_{s=0}^{\infty}\pi^{\hat{g}}(s) \left(c\left(\left\lfloor \frac{1}{2}s\right\rfloor\right)
+c\left(\left\lceil \frac{1}{2}s\right\rceil\right) \right) a.s.
\label{eq:inf:lim}
\end{align}
where $\pi^{\hat{g}}(s)$ is the stationary distribution of $\{S_t=S^{\hat{g}}_t, t\geq T_0+1\}$. The second equality in (\ref{eq:inf:lim}) holds because $T_0<\infty \, a.s.$;
the last equality in (\ref{eq:inf:lim}) follows by the Ergodic theorem for irreducible positive recurrent Markov chains \cite[chap. 3]{bremaud1999markov}.
\\
Since the sum $\frac{1}{T}\sum_{t=0}^{T-1}
\left(c\left(X^{1,\hat{g}}_t\right)+c\left(X^{2,\hat{g}}_t\right)\right)$ converges $a.s.$, 
from Corollary \ref{cor:infcostfinite} we have
\begin{align}
J^{\hat{g}}(\pi^1_0,\pi^2_0) 
= &\lim_{T\rightarrow\infty} \frac{1}{T} \sum_{t=0}^{T-1}\left(c\left(X^{1,\hat{g}}\right)+c\left(X^{2,\hat{g}}\right)\right) \nonumber\\
= & \sum_{s=0}^{\infty}\pi^{\hat{g}}(s) \left(c\left(\left\lfloor \frac{1}{2}s\right\rfloor\right)
+c\left(\left\lceil \frac{1}{2}s\right\rceil\right) \right).
\label{eq:inf:limE}
\end{align}
Since the right hand side of equation \eqref{eq:inf:limE} does not depend on the initial PMFs $\pi^1_0,\pi^2_0$, we obtain
\begin{align}
J^{\hat{g}}(\pi^1_0,\pi^2_0) = J^{\hat{g}}(0,0).
\label{eq:inf:costis00}
\end{align}
Combining (\ref{eq:inf:costis00}) and Lemma \ref{lm:cen} we get
\begin{align}
J^{\hat{g}}(\pi^1_0,\pi^2_0) = J^{\hat{g}}(0,0) \leq \inf_{g \in \mathcal{G}_d}J^{g}(\pi^1_0,\pi^2_0).
\end{align}
Thus, $\hat{g}$ is an optimal routing policy for the infinite horizon problem.

\end{proof}

\section{Discussion and Conclusion}
\label{sec:diss}

Based on the results established in Sections \ref{sec:SR}-\ref{sec:infinite}, we now discuss and answer the questions posed in Section \ref{sec:intro}.

Controllers $C_1$ and $C_2$ communicate with one another through their control actions; thus, each controller's information depends on the decision rule/routing policy of the other controller. Therefore, the queueing system considered in this paper has non-classical information structure \cite{witsenhausen1971separation}. A key feature of the system's information structure is that at each time instant each controller's information consists of one component that is common knowledge between $C_1$ and $C_2$ and another component that is its own private information. 
The presence of common information allows us to use the common information approach, developed in \cite{nayyar2013decentralized}, along with specific features of our model to identify an information state/sufficient statistic for the finite and infinite horizon optimization problem. The identification/discovery of an appropriate information state proceeds in two steps: 
In the first step we use the common information approach (in particular \cite{mahajan2012optimal}) to identify the general form of an information state (namely $\left(\overline{X}^i_t,\overline{\Pi}^1_t,\overline{\Pi}^2_t\right)$) for controller $C_i,i=1,2$. (and the corresponding structure of an optimal policy, Properties \ref{P:sep2}).
In the second step we take advantage of the features of our system to further refine/simplify the information state; we discover a simpler form of information state, namely, 
$\left( \overline{X}^i_t, \;\left\{\overline{UB}^{j}_t,\overline{LB}^{j}_t\right\}_{j=1,2} \right)$ for controller $C_i,i=1,2$.
The component $\left\{\overline{UB}^{j}_t,\overline{LB}^{j}_t\right\}_{j=1,2}$ of the above information state describes the common information between controllers $C_1$ and $C_2$ at time $t, t=1,2,\dots$.

Using this common information we established an optimal signaling strategy that is described by the threshold policy $\hat{g}$ specified in Section \ref{sec:prelim}.

The update of $\left\{\overline{UB}^{j}_t,\overline{LB}^{j}_t\right\}_{j=1,2}$ is described by \eqref{eq:UBLBupdate00}-\eqref{eq:UBLBupdate10} and explicitly depends on the signaling policy $\hat{g}$.
Specifically, if a customer is sent from $Q_i$ to $Q_j\,(i \neq j)$ at time $t$ 
the lower bound on the queue length of $Q_i$ increases because
both controllers know that 
the length of $Q_i$ is above the threshold $TH_t$ at the time of routing;
if no customer is sent from $Q_i$ to $Q_j$ at time $t$, 
the upper bound on the length of $Q_i$ decreases because
both controllers know that 
the length of $Q_i$ is below the threshold $TH_t$ at the time of routing.
The update of common information incorporates the information about a controller's private information transmitted to the other controller through signaling.

The signaling policy $\hat{g}$ communicates information in such a way that eventually the difference between the upper bound and the lower bound on the queue lengths is no more than one. Thus, signaling through $\hat{g}$ results in a balanced queueing system.

\appendix
\section{Proofs of the Results in Section \ref{sec:prelim}}
\label{app:prelim}

\begin{proof}[Proof of Lemma \ref{lm:bounds}]
Since there is one possible arrival to any queue and one possible departure from any queue at each time instant, 
(\ref{eq:UBLBbarupdate}) holds.

When $(U^{1,\hat{g}}_t,U^{2,\hat{g}}_t)=(0,0)$, both $\overline{X}^{1,\hat{g}}_t$ and $\overline{X}^{2,\hat{g}}_t$ are below the threshold and no customers are routed form any queue. 
Therefore, the upper bound of the queue lengths at $t+1$ is
\begin{align}
UB^{\hat{g}}_{t+1} 
=& \left\lceil
TH_t
\right \rceil-1.
\label{eq:UB00}
\end{align}
Moreover, the lower bound of the queue lengths at $t+1$ is the same as the lower bound of $\overline{X}^{1,\hat{g}}_t,\overline{X}^{2,\hat{g}}_t$. That is,
\begin{align}
LB^{\hat{g}}_{t+1} =&\overline{LB}^{\hat{g}}_t.
\label{eq:LB00}
\end{align}
When $(U^{1,\hat{g}}_t,U^{2,\hat{g}}_t)=(1,1)$, both $\overline{X}^{1,\hat{g}}_t$ and $\overline{X}^{2,\hat{g}}_t$ are greater than or equal to the threshold. 
Since the routing only exchanges two customers between the two queues, the queue lengths remain the same as the queue lengths before routing.
As a result, the upper bound and lower bound of the queue lengths at $t+1$ are given by
\begin{align}
UB^{\hat{g}}_{t+1} =& \overline{UB}^{\hat{g}}_t.\label{eq:UB11}\\
LB^{\hat{g}}_{t+1} =&\left\lceil TH_t \right\rceil. \label{eq:LB11}
\end{align}
When $(U^{i,\hat{g}}_t,U^{j,\hat{g}}_t)=(1,0), i\neq j$, $\overline{X}^{i,\hat{g}}_t$ is greater than or equal to the threshold; 
$\overline{X}^{j,\hat{g}}_t$ is below the threshold.
Since one customer is routed from $Q_i$ to $Q_j$,
\begin{align}
&X^{i,\hat{g}}_{t+1} = \overline{X}^{i,\hat{g}}_t-1,\\
&X^{j,\hat{g}}_{t+1} = \overline{X}^{j,\hat{g}}_t+1.
\end{align} 
Therefore, the upper bound of the queue lengths at $t+1$ becomes
\begin{align}
UB^{\hat{g}}_{t+1} 
=& \max \left\lbrace
\overline{UB}^{i,\hat{g}}_t-1
,\left\lceil
TH_t
\right \rceil-1 + 1\right\rbrace\nonumber\\
=& \max \left\lbrace
\overline{UB}^{i,\hat{g}}_t-1
,\left\lceil
TH_t
\right \rceil\right\rbrace, \label{eq:UB10}
\end{align}
and lower bound of the queue lengths at $t+1$ is given by
\begin{align}
LB^{\hat{g}}_{t+1} 
=& \min \left\lbrace
\left\lceil
TH_t
\right \rceil -1 
,\overline{LB}^{j,\hat{g}}_t+1
\right\rbrace.
 \label{eq:LB10}
\end{align}
\end{proof}

\section{Proofs of the Results in Section \ref{sec:finite}}
\label{app:finite}

\begin{proof}[Proof of Lemma \ref{lm:stoless}]
The proof is done by induction.\\
At time $t=0$, $X^{1,\hat{g}}_0+X^{2,\hat{g}}_0 = X^{1,g}_0+X^{2,g}_0=x_0$.\\
Suppose the lemma is true at time $t$.\\
At time $t+1$, from the system dynamics (\ref{Model:dynamic1})-(\ref{Model:dynamic3}) we get, for any $g$,
\begin{align}
&X^{1,g}_{t+1}+X^{2,g}_{t+1}\nonumber\\
=& \left(X^{1,g}_t - D^1_t \right)^+ +\left(X^{2,g}_t - D^2_t \right)^+ +A^1_t+A^2_t.
\label{pflmstoless:eq1}
\end{align}
Therefore, it suffices to show that
\begin{align}
\left(X^{1,\hat{g}}_t - D^1_t \right)^+ +\left(X^{2,\hat{g}}_t - D^2_t \right)^+ 
\leq_{st} & \left(X^{1,g}_t - D^1_t \right)^+ +\left(X^{2,g}_t - D^2_t \right)^+.
\label{pflmstoless:eqnoA}
\end{align}
Consider any realization $(X^{1,g}_t,X^{2,g}_t) = (x^1,x^2)$.
\\
If $x^1,x^2 >0$, then $\left\lfloor \frac{1}{2}(x^1+x^2 )\right\rfloor,\left\lceil \frac{1}{2}(x^1+x^2)\right\rceil >0$.
Therefore,
\begin{align}
&\left(X^{1,g}_t - D^1_t \right)^+ +\left(X^{2,g}_t - D^2_t \right)^+\nonumber\\
=&x^1+x^2- D^1_t- D^2_t\nonumber\\
=&\left(\left\lfloor \frac{1}{2}(x^1+x^2 )\right\rfloor- D^1_t\right)^+
 +\left(\left\lceil \frac{1}{2}(x^1+x^2)\right\rceil- D^2_t\right)^+.
 \label{pflmstoless:eq1}
\end{align}
If $x^i=0$ and $x^j \geq 2$ ($i\neq j$), then $\left\lfloor \frac{1}{2}(x^1+x^2 )\right\rfloor>0$ and $\left\lceil \frac{1}{2}(x^1+x^2)\right\rceil >0$.
Therefore, 
\begin{align}
&\left(X^{1,g}_t - D^1_t \right)^+ +\left(X^{2,g}_t - D^2_t \right)^+\nonumber\\
=&x^j- D^j_t\nonumber\\
\geq &x^1+x^2- D^1_t- D^2_t\nonumber\\
=&\left(\left\lfloor \frac{1}{2}(x^1+x^2 )\right\rfloor- D^1_t\right)^+
 +\left(\left\lceil \frac{1}{2}(x^1+x^2)\right\rceil- D^2_t\right)^+.
 \label{pflmstoless:eq2}
\end{align}
If $x^i=0$ and $x^j =1 $ ($i\neq j$), then $\left\lfloor \frac{1}{2}(x^1+x^2 )\right\rfloor=0$ and $\left\lceil \frac{1}{2}(x^1+x^2)\right\rceil =1$.
Therefore, 
\begin{align}
&\left(X^{1,g}_t - D^1_t \right)^+ +\left(X^{2,g}_t - D^2_t \right)^+\nonumber\\
=&1- D^j_t\nonumber\\
\geq_{st} &1-D^2_t\nonumber\\
=&\left(\left\lfloor \frac{1}{2}(x^1+x^2 )\right\rfloor- D^1_t\right)^+
 +\left(\left\lceil \frac{1}{2}(x^1+x^2)\right\rceil- D^2_t\right)^+,
  \label{pflmstoless:eq3}
\end{align}
If $x^1,x^2 =0$, then $\left\lfloor \frac{1}{2}(x^1+x^2 )\right\rfloor,\left\lceil \frac{1}{2}(x^1+x^2)\right\rceil =0$.
Therefore, 
\begin{align}
&\left(X^{1,g}_t - D^1_t \right)^+ +\left(X^{2,g}_t - D^2_t \right)^+\nonumber\\
=&0\nonumber\\
=&\left(\left\lfloor \frac{1}{2}(x^1+x^2 )\right\rfloor- D^1_t\right)^+
 +\left(\left\lceil \frac{1}{2}(x^1+x^2)\right\rceil- D^2_t\right)^+.
  \label{pflmstoless:eq4}
\end{align}
As a result of (\ref{pflmstoless:eq1})-(\ref{pflmstoless:eq4}), we obtain
\begin{align}
&\left(X^{1,g}_t - D^1_t \right)^+ +\left(X^{2,g}_t - D^2_t \right)^+\nonumber\\
\geq_{st} & 
\left(\left\lfloor \frac{1}{2}(X^{1,g}_t+X^{2,g}_t )\right\rfloor- D^1_t\right)^+
+\left(\left\lceil \frac{1}{2}(X^{1,g}_t+X^{2,g}_t)\right\rceil- D^2_t\right)^+.
\label{pflmstoless:eq2}
\end{align}
Then, from (\ref{pflmstoless:eq2}), the induction hypothesis and Corollary \ref{cor:bal} we obtain
\begin{align}
&\left(X^{1,g}_t - D^1_t \right)^+ +\left(X^{2,g}_t - D^2_t \right)^+\nonumber\\
\geq_{st} 
&\left(\left\lfloor \frac{1}{2}(X^{1,g}_t+X^{2,g}_t )\right\rfloor- D^1_t\right)^+
+\left(\left\lceil \frac{1}{2}(X^{1,g}_t+X^{2,g}_t)\right\rceil- D^2_t\right)^+\nonumber\\
\geq_{st} 
&\left(\left\lfloor \frac{1}{2}(X^{1,\hat{g}}_t+X^{2,\hat{g}}_t )\right\rfloor- D^1_t\right)^+
+\left(\left\lceil \frac{1}{2}(X^{1,\hat{g}}_t+X^{2,\hat{g}}_t)\right\rceil- D^2_t\right)^+ \nonumber\\
= & \left(\min(X^{1,\hat{g}}_t,X^{2,\hat{g}}_t)- D^1_t\right)^++\left(\max(X^{1,\hat{g}}_t,X^{2,\hat{g}}_t)- D^2_t\right)^+\nonumber\\
\geq_{st} & \left(X^{1,\hat{g}}_t- D^1_t\right)^+ +\left(X^{2,\hat{g}}_t- D^2_t\right)^+.
\label{pflmstoless:eq3}
\end{align}
The first and second stochastic inequalities in \eqref{pflmstoless:eq3} follow from 
 \eqref{pflmstoless:eq2} and the induction hypothesis, respectively.
The equality in \eqref{pflmstoless:eq3} follows from Corollary \ref{cor:bal}.
The last stochastic inequality in \eqref{pflmstoless:eq3} is true because $D^1_t, D^2_t$ are i.i.d. and independent of $X^{1,\hat{g}}_t, X^{2,\hat{g}}_t$.
\\
Thus, inequality (\ref{pflmstoless:eqnoA}) is true, and the proof of the lemma is complete.
\end{proof}

\section{Proofs of the Results Associated with Step 1 of the Proof of Theorem \ref{thm:inf}}
\label{app:infinite1}

\begin{proof}[Proof of Lemma \ref{lm:inf:stobdd}]
The proof is done by induction.
At $t=0$, (\ref{eq:stoeq}), (\ref{eq:stosumbdd}) and (\ref{eq:stomaxbdd}) hold if we let 
$Y^{i}_0 = X^{i,g_0}_0$  for $i=1,2$. \\
Assume the assertion of this lemma is true at time $t$; we want to show that the assertion is also true at time $t+1$.
\\
For that matter we claim the following.
\\
\underline{Claim 1}
\begin{align}
X^{1,\hat{g}}_{t+1}+X^{2,\hat{g}}_{t+1} = \overline{X}^{1,\hat{g}}_t+\overline{X}^{2,\hat{g}}_t \quad a.s. ,
\label{eq:xghatsum}
\\
\max_i\left(X^{i,\hat{g}}_{t+1}\right) \leq \max_i\left(\overline{X}^{i,\hat{g}}_t\right) \quad a.s.
\label{eq:xghatmax}
\end{align}
\underline{Claim 2}
\\
There exists $Y^{i}_{t+1}, i=1,2$ such that
\begin{align}
 \mathbf{P}\left(Y^{i}_{t+1} = y_{t+1} | Y^{i}_{0:t}=y_{0:t}\right) 
= &\mathbf{P}\left( X^{i,g0}_{t+1} =y_{t+1}| X^{i,g_0}_{0:t}=y_{0:t}\right) 
\text{ for all }y_{0:t},
\label{eq:yxg0eq}
\\
\overline{X}^{1,\hat{g}}_t+\overline{X}^{2,\hat{g}}_t 
\leq&  Y^{1}_{t+1}+Y^{2}_{t+1} \quad a.s.,  \label{eq:xbarysum}
\\
\max_i\left(\overline{X}^{i,\hat{g}}_t\right) 
\leq &\max_i\left(Y^{i}_{t+1}\right) \quad a.s.\label{eq:xbarymax}
\end{align}
We assume the above claims to be true and prove them after the completion of the proof of the induction step.
\\
For all $y_{0:t+1}$, from (\ref{eq:yxg0eq}) and the induction hypothesis for (\ref{eq:stoeq}) we get for $i=1,2$
\begin{align}
& \mathbf{P}\left(Y^{i}_{0:t+1} = y_{0:t+1}\right) \nonumber\\
=& \mathbf{P}\left(Y^{i}_{t+1} = y_{t+1}| Y^{i}_{0:t}=y_{0:t}\right) 
	 \mathbf{P}\left(Y^{i}_{t}=y_t,\dots, Y^{i}_{0}=y_0\right) \nonumber\\
= &\mathbf{P}\left( X^{i,g0}_{t+1} =y_{t+1}| X^{i,g_0}_{0:t}=y_{0:t}\right)  	
	\mathbf{P}\left( X^{i,g_0}_{0:t}=y_{0:t}\right)  \nonumber\\
=&\mathbf{P}\left( X^{i,g0}_{0:t+1} =y_{0:t+1}\right).
\label{eq:indeq}
\end{align}
From (\ref{eq:xghatsum}) and (\ref{eq:xbarysum}) we obtain
\begin{align}
X^{1,\hat{g}}_{t+1}+X^{2,\hat{g}}_{t+1}
= &  \overline{X}^{1,\hat{g}}_t+\overline{X}^{2,\hat{g}}_t \nonumber\\
\leq &Y^{1}_{t+1}+Y^{2}_{t+1}  \quad a.s. \label{eq:indsum}
\end{align}
Furthermore, combination of (\ref{eq:xghatmax}) and (\ref{eq:xbarymax}) gives
\begin{align}
\max_i\left(X^{i,\hat{g}}_{t+1}\right) 
\leq \max_i\left(\overline{X}^{i,\hat{g}}_t\right)
= \max_i\left(Y^{i}_{t+1}\right)  \quad a.s.\label{eq:indmax}
\end{align}
Therefore, the assertions (\ref{eq:stoeq}), (\ref{eq:stosumbdd}) and (\ref{eq:stomaxbdd}) of the lemma are true at $t+1$
 by (\ref{eq:indeq}), (\ref{eq:indsum}) and (\ref{eq:indmax}), respectively.
\\
We now prove claims 1 and 2.
\\
\underline{Proof of Claim 1}\\
From the system dynamics \eqref{Model:dynamic1}-\eqref{Model:dynamic2}
\begin{align}
X^{1,\hat{g}}_{t+1} = \overline{X}^{i,\hat{g}}_t-U^{i,\hat{g}}_t+U^{j,\hat{g}}_t,
\label{eq:pflmstobdd:dyn1}
\\
X^{2,\hat{g}}_{t+1} = \overline{X}^{i,\hat{g}}_t-U^{i,\hat{g}}_t+U^{j,\hat{g}}_t.
\label{eq:pflmstobdd:dyn2}
\end{align}
Therefore, (\ref{eq:xghatsum}) follows by summing (\ref{eq:pflmstobdd:dyn1}) and (\ref{eq:pflmstobdd:dyn2}).
\\
For (\ref{eq:xghatmax}), consider $X^{1,\hat{g}}_{t+1}$ ( the case of $X^{2,\hat{g}}_{t+1}$ follows from similar arguments).
\\
When $U^{2,\hat{g}}_t = 0$, 
\begin{align}
X^{1,\hat{g}}_{t+1} = \overline{X}^{1,\hat{g}}_t-U^{1,\hat{g}}_t \leq \max_i\left(\overline{X}^{i,\hat{g}}_t\right).
\label{eq:xghatmax1}
\end{align}
When $U^{1,\hat{g}}_t =U^{2,\hat{g}}_t = 1$, 
\begin{align}
X^{1,\hat{g}}_{t+1} = \overline{X}^{1,\hat{g}}_t \leq \max_i\left(\overline{X}^{i,\hat{g}}_t\right).\label{eq:xghatmax2}
\end{align}
When $U^{1,\hat{g}}_t = 0, U^{2,\hat{g}}_t = 1$, $\overline{X}^{1,\hat{g}}_t$ is less than the threshold and 
$\overline{X}^{2,\hat{g}}_t$ is greater than or equal to the threshold. Therefore, by \eqref{eq:pflmstobdd:dyn1},
\begin{align}
X^{1,\hat{g}}_{t+1} = \overline{X}^{1,\hat{g}}_t +1 
\leq & \left\lceil
TH_t
\right \rceil \nonumber\\
\leq & \overline{X}^{2,\hat{g}}_t
\leq  \max_i\left(\overline{X}^{i,\hat{g}}_t\right). \label{eq:xghatmax3}
\end{align}
Therefore, (\ref{eq:xghatmax}) follows from (\ref{eq:xghatmax1})-(\ref{eq:xghatmax3}).
\\
\underline{Proof of Claim 2}\\
We set
\begin{align}
Y^i_{t+1} := \left( Y^{i}_t -\tilde{D}^i_t\right)^+ +\tilde{A}^i_t
\label{eq:ytp1}
\end{align}
where $Y^i_t$ satisfy the induction hypothesis, and $\tilde{A}^i_t, \tilde{D}^i_t, i=1,2$ are 
specified as follows.
Let
\begin{align}
M_x =& \text{argmax}_i\{X^{i,\hat{g}}_t\}, \quad 
m_x = \text{argmin}_i\{X^{i,\hat{g}}_t\}\\
M_y =& \text{argmax}_i\{Y^i_t\}, \quad 
m_y = \text{argmin}_i\{Y^i_t\},
\end{align}
where $M_x=1, m_x = 2$ (resp. $M_y=1, m_y = 2$) when $\{X^{1,\hat{g}}_t=X^{2,\hat{g}}_t\}$ (resp. $\{Y^1_t=Y^2_t\}$); define
\begin{align}
\left(\tilde{A}^{M_y}_t,\tilde{D}^{M_y}_t,\tilde{A}^{m_y}_t,\tilde{D}^{m_y}_t\right)
:=
\left\lbrace
\begin{array}{ll}
\left(A^{M_x}_t,D^{m_x}_t,A^{m_x}_t,D^{M_x}_t\right)  &\text{ in case 1},
\\
\left(A^{M_x}_t,D^{M_x}_t,A^{m_x}_t,D^{m_x}_t\right)  &\text{ in case 2},
\end{array}
\right.
\label{eq:ADtilde}
\end{align}
where the two cases are :\\
Case 1: 
$\{Y^{M_y}_{t}-1 = X^{M_x,\hat{g}}_{t}=X^{m_x,\hat{g}}_{t}$ and
$\left(A^{M_x}_t,D^{M_x}_t,A^{m_x}_t,D^{m_x}_t\right) = (0,1,1,0) \text{ or } (0,0,1,1) \}$.
\\
Case 2: All other instances.
\\

\underline{Assertion}: The random variables $Y^1_{t+1},Y^2_{t+1}$, defined by 
\eqref{eq:ytp1}-\eqref{eq:ADtilde} satisfy \eqref{eq:yxg0eq}-\eqref{eq:xbarymax}.

As the proof of this assertion is long, we first provide a sketch of its proof and then we provide a full proof.
\\
\underline{Sketch of the proof of the assertion}
\begin{itemize}
\item Equation \eqref{eq:ADtilde} implies the following:
In case 2 we associate the arrival to and the departure from the longer queue $M_x$ to those of the longer queue $M_y$,
i.e. we set $\tilde{A}^{M_y}_t = A^{M_x}_t, \tilde{D}^{M_y}_t = D^{M_x}_t$.
We do the same for the shorter queue $m_x, m_y$, i.e. $\tilde{A}^{m_y}_t = A^{m_x}_t, \tilde{D}^{m_y}_t = D^{m_x}_t$.
\\
In case 1, we have the same association for the arrivals as in case 2, that is $\tilde{A}^{M_y}_t = A^{M_x}_t, \tilde{A}^{m_y}_t = A^{m_x}_t$, but we reverse the association of the departures, that is $\tilde{D}^{M_y}_t = D^{m_x}_t, \tilde{D}^{m_y}_t = D^{M_x}_t$.
Therefore the arrivals $\tilde{A}^{i}_t$, and departures $\tilde{D}^{i}_t$, have the same distribution as the original $A^{i}_t, D^{i}_t$, respectively, $i=1,2$.
Then (\ref{eq:yxg0eq}) follows from \eqref{eq:ytp1}.

\item To establish \eqref{eq:xbarysum}, we note that,
because of \eqref{eq:ADtilde}, the sum of arrivals to (respectively, departures from)
queues $M_y$ and $m_y$ equals to the sum of arrivals to (respectively, departures from) queues $M_x$ and $m_x$.
\\
When $X^{i,\hat{g}}_t,Y^i_t \neq 0$, $i=1,2$, the function $(x-d)^+ + a$ is linear $x$, as $(x-d)^+ + a = x-d+a$. 
Then from \eqref{eq:ytp1}, \eqref{eq:ADtilde} and the induction hypothesis we obtain
\begin{align}
& Y^1_{t+1}+Y^2_{t+1} - \overline{X}^{1,\hat{g}}_t-\overline{X}^{2,\hat{g}}_t \nonumber\\
= &Y^1_t+Y^2_t - X^{1,\hat{g}}_t-X^{2,\hat{g}}_t \geq 0
\end{align}
and this establish (\ref{eq:xbarysum}) when $X^{i,\hat{g}}_t,Y^i_t \neq 0$, $i=1,2$.
In the full proof of the assertion, we show that show that (\ref{eq:xbarysum}) is also true when $X^{i,\hat{g}}_t,Y^i_t$ are not all non-zero.
\item To establish (\ref{eq:xbarymax}) we consider the maximum of the queue lengths.
In case 2, we show that \eqref{eq:ytp1}-\eqref{eq:ADtilde} ensure that
\begin{align}
&Y^{M_y}_{t+1} \geq \overline{X}^{M_x,\hat{g}}_t, 
\label{eq:MygeqMx2}
\\
&\max\left( Y^{M_y}_{t+1},Y^{m_y}_{t+1}\right) \geq \overline{X}^{m_x,\hat{g}}_t;
\label{eq:Mymygeqmx2}
\end{align}
then \eqref{eq:xbarymax} follows from \eqref{eq:MygeqMx2}-\eqref{eq:Mymygeqmx2}.
\\
In case 1 \eqref{eq:xbarymax} is verified by direct computation in the full proof.

\end{itemize}
\underline{Proof of the assertion}\\
For all $y_{0:t}$, we denote by $E_{y_{0:t}}$ the event $\{Y^{i}_{0:t}=y_{0:t} \}$.
\\
Let $\tilde{Z}_t=\left(\tilde{A}^{M_y}_t,\tilde{D}^{M_y}_t,\tilde{A}^{m_y}_t,\tilde{D}^{m_y}_t\right)$,
then for any realization $z_t \in \{0,1\}^4 $ of $\tilde{Z}_t$ we have
\begin{align}
& \mathbf{P}\left(\tilde{Z}_t =z_t | E_{y_{0:t}}\right) \nonumber\\
=& \mathbf{P}\left(\tilde{Z}_t =z_t, \text{case 1}| E_{y_{0:t}}\right)+ \mathbf{P}\left(\tilde{Z}_t =z_t, \text{case 2}| E_{y_{0:t}}\right).
\label{eq:ZtgivenEy}
\end{align}
When $z_t \neq (0,1,1,0)$ or $(0,0,1,1)$, we get 
\begin{align}
\mathbf{P}\left(\tilde{Z}_t =z_t, \text{case 1}| E_{y_{0:t}}\right)=0,
\label{eq:Ztneqcase1}
\end{align} 
and
\begin{align}
 & \mathbf{P}\left(\tilde{Z}_t =z_t, \text{case 2}| E_{y_{0:t}}\right)\nonumber\\
= & \mathbf{P}\left(\left(A^{M_x}_t,D^{M_x}_t,A^{m_x}_t,D^{m_x}_t\right)  =z_t | E_{y_{0:t}}\right)\nonumber\\
= & \mathbf{P}\left(\left(A^{1}_t,D^{1}_t,A^{2}_t,D^{2}_t\right)  =z_t\right),
\label{eq:Ztneqcase2}
\end{align}
where the last equality in (\ref{eq:Ztneqcase2}) holds because the random variables $A^{M_x}_t,D^{M_x}_t,A^{m_x}_t,D^{m_x}_t$ are independent of $Y_0,Y_1,\dots,Y_t$ and have the same distribution as $A^{1}_t,D^{1}_t,A^{2}_t,D^{2}_t$.
\\
Therefore, combining (\ref{eq:Ztneqcase1}) and (\ref{eq:Ztneqcase2}) we obtain for $z_t \neq (0,1,1,0)$ or $(0,0,1,1)$
\begin{align}
& \mathbf{P}\left(\tilde{Z}_t =z_t | E_{y_{0:t}}\right) = \mathbf{P}\left(\left(A^{1}_t,D^{1}_t,A^{2}_t,D^{2}_t\right)  =z_t\right)
\label{eq:Ztneq}
\end{align}
When $z_t = (0,1,1,0)$ or $(0,0,1,1)$, let $E$ denote the event $\{Y^{M_y}_{t}-1 = X^{M_x,\hat{g}}_{t}=X^{m_x,\hat{g}}_{t}\}$;
then we obtain 
\begin{align}
& \mathbf{P}\left(\tilde{Z}_t =z_t, \text{case 1}| E_{y_{0:t}}\right) \nonumber\\
= & \mathbf{P}\left(\left(A^{M_x}_t,D^{m_x}_t,A^{m_x}_t,D^{M_x}_t\right)  =z_t, E | E_{y_{0:t}}\right)\nonumber\\
= & \mathbf{P}\left(\left(A^{1}_t,D^{2}_t,A^{2}_t,D^{1}_t\right)  =z_t\right) \mathbf{P}\left(E | E_{y_{0:t}}\right),
\label{eq:Zteqcase1}
\end{align} 
and
\begin{align}
 & \mathbf{P}\left(\tilde{Z}_t =z_t, \text{case 2}| E_{y_{0:t}}\right)\nonumber\\
= & \mathbf{P}\left(\left(A^{M_x}_t,D^{M_x}_t,A^{m_x}_t,D^{m_x}_t\right)  =z_t, E^c | E_{y_{0:t}}\right)\nonumber\\
= & \mathbf{P}\left(\left(A^{1}_t,D^{1}_t,A^{2}_t,D^{2}_t\right)  =z_t\right) \mathbf{P}\left(E^c | E_{y_{0:t}}\right),
\label{eq:Zteqcase2}
\end{align}
where the last equality in (\ref{eq:Zteqcase1}) and (\ref{eq:Zteqcase2}) follow by the fact that the random variables
$A^{M_x}_t,D^{M_x}_t,A^{m_x}_t,D^{m_x}_t$ are independent of $Y_0,Y_1,\dots,Y_t$ (hence, the event $E$ which is generated by $Y_0,Y_1,\dots,Y_t$) and have the same distribution as $A^{1}_t,D^{1}_t,A^{2}_t,D^{2}_t$.
\\
Therefore, combining (\ref{eq:Zteqcase1}) and (\ref{eq:Zteqcase2}) we obtain for $z_t = (0,1,1,0)$ or $(0,0,1,1)$
\begin{align}
& \mathbf{P}\left(\tilde{Z}_t =z_t | E_{y_{0:t}}\right) \nonumber\\
= &\mathbf{P}\left(\left(A^{1}_t,D^{2}_t,A^{2}_t,D^{1}_t\right)  =z_t\right) \mathbf{P}\left(E | E_{y_{0:t}}\right)\nonumber\\
 & +\mathbf{P}\left(\left(A^{1}_t,D^{1}_t,A^{2}_t,D^{2}_t\right)  =z_t\right) \mathbf{P}\left(E^c | E_{y_{0:t}}\right)\nonumber\\
= & \mathbf{P}\left(\left(A^{1}_t,D^{1}_t,A^{2}_t,D^{2}_t\right)  =z_t\right),
\label{eq:Zteq}
\end{align}
where the last equality in (\ref{eq:Zteq}) is true because $A^{1}_t,D^{1}_t,A^{2}_t,D^{2}_t$ are independent and $D^1_t$ has the same distribution as $D^{2}_t$.
\\
As a result of (\ref{eq:Ztneq}) and (\ref{eq:Zteq}), for any $z_t \in \{0,1\}^4 $ we have
\begin{align}
& \mathbf{P}\left(\tilde{Z}_t =z_t | E_{y_{0:t}}\right) =\mathbf{P}\left(\left(A^{1}_t,D^{1}_t,A^{2}_t,D^{2}_t\right)  =z_t\right).
\label{eq:ZteqE}
\end{align}
Now consider any $y_{0:t+1}$. By \eqref{eq:ZteqE} we have for $i=M_y$ or $m_y$
\begin{align}
 &\mathbf{P}\left(Y^{i}_{t+1} = y_{t+1} | E_{y_{0:t}}\right) \nonumber\\
= &\mathbf{P}\left(\left(y^{i}_{t} - \tilde{D}^i_t \right)^+  +\tilde{A}^i_t = y_{t+1} | E_{y_{0:t}}\right) \nonumber\\
= &\mathbf{P}\left(\left(y^{i}_{t} - D^i_t \right)^+  +A^i_t = y_{t+1}\right) \nonumber\\
= &\mathbf{P}\left( X^{i,g0}_{t+1} =y_{t+1}| X^{i,g_0}_{0:t}=y_{0:t}\right) .
\end{align}
which is (\ref{eq:yxg0eq}).\\
Now consider the sum $Y^1_{t+1}+Y^2_{t+1}$.
\\
From \eqref{eq:ADtilde}, we know that 
\begin{align}
&\tilde{A}^{M_y}_t+\tilde{A}^{m_y}_t= A^{M_x}_t+A^{m_x}_t \quad a.s., \label{eq:AYX}\\
&\tilde{D}^{M_y}_t+\tilde{D}^{m_y}_t= D^{M_x}_t+D^{m_x}_t \quad a.s. \label{eq:DYX}
\end{align}
Therefore, (\ref{eq:AYX}) implies
\begin{align}
&Y^1_{t+1}+Y^2_{t+1} - \overline{X}^{1,\hat{g}}_{t+1}-\overline{X}^{1,\hat{g}}_{t+1} \nonumber\\
=&
\left(Y^{M_y}_{t} - \tilde{D}^{M_y}_t \right)^++\left(Y^{m_y}_{t} - \tilde{D}^{m_y}_t \right)^+ \nonumber\\
&- \left(X^{M_x,\hat{g}}_{t} - D^{M_x}_t \right)^+ -\left(X^{m_x,\hat{g}}_{t} - D^{m_x}_t \right)^+.
\label{eq:sumYXbar}
\end{align}
We proceed to show that the right hand side of (\ref{eq:sumYXbar}) is positive.
From the induction hypothesis for \eqref{eq:stomaxbdd}-\eqref{eq:stosumbdd} we have
\begin{align}
&Y^{m_y}_{t}+Y^{M_y}_{t} \geq X^{m_x,\hat{g}}_{t}+X^{M_x,\hat{g}}_{t}\quad a.s. ,
\label{eq:MyMxsum}
\\
&Y^{M_y}_{t} \geq X^{M_x,\hat{g}}_{t}\quad a.s.
\label{eq:MyMx}
\end{align}
There are three possibilities: $\{Y^{M_y}_{t} = X^{M_x,\hat{g}}_{t}\}$, 
 $\{Y^{M_y}_{t} > X^{M_x,\hat{g}}_{t}, X^{m_x,\hat{g}}_{t} = 0\}$ and
 $\{Y^{M_y}_{t} > X^{M_x,\hat{g}}_{t}, X^{m_x}_{t} > 0\}$.
\\
First consider $\{Y^{M_y}_{t} = X^{M_x,\hat{g}}_{t}\}$. 
By \eqref{eq:MyMxsum} we have
\begin{align}
Y^{m_y}_{t} \geq X^{m_x,\hat{g}}_{t} \quad a.s.
\label{eq:mymx}
\end{align}
Note that $\{Y^{M_y}_{t} = X^{M_x,\hat{g}}_{t}\}$ belongs to case 2 in \eqref{eq:ADtilde}.
From case 2 of \eqref{eq:ADtilde} we also know that
\begin{align}
D^{M_x}_t=\tilde{D}^{M_y}_t, \quad D^{m_x}_t=\tilde{D}^{m_y}_t.
\label{eq:DMDm}
\end{align}
Then, because of \eqref{eq:MyMx}-\eqref{eq:DMDm} we get
\begin{align}
&\left(X^{M_x,\hat{g}}_{t} - D^{M_x}_t \right)^+ +\left(X^{m_x,\hat{g}}_{t} - D^{m_x}_t \right)^+ \nonumber\\
\leq 
&\left(Y^{M_y}_{t} - D^{M_x}_t \right)^+ +\left(Y^{m_y}_{t} - D^{m_x}_t \right)^+ \nonumber\\
= &\left(Y^{M_y}_{t} - \tilde{D}^{M_y}_t \right)^++\left(Y^{m_y}_{t} - \tilde{D}^{m_y}_t \right)^+ \quad a.s.
\label{eq:YeqXsumeq}
\end{align}
If $Y^{M_y}_{t} > X^{M_x,\hat{g}}_{t} $ and $X^{m_x,\hat{g}}_{t} = 0$
\begin{align}
&\left(X^{M_x,\hat{g}}_{t} - D^{M_x}_t \right)^+ +\left(X^{m_x,\hat{g}}_{t} - D^{m_x}_t \right)^+
\nonumber\\
= & \left(X^{M_x,\hat{g}}_{t} - D^{M_x}_t \right)^+
\nonumber\\
\leq & X^{M_x,\hat{g}}_{t}
\leq  Y^{M_y}_{t}-1 \nonumber\\
\leq &\left(Y^{M_y}_{t} - \tilde{D}^{M_y}_t \right)^++\left(Y^{m_y}_{t} - \tilde{D}^{m_y}_t \right)^+
\end{align}
If $Y^{M_y}_{t} > X^{M_x,\hat{g}}_{t} $ and $X^{m_x}_{t} > 0$, then 
\begin{align}
&\left(X^{M_x,\hat{g}}_{t} - D^{M_x}_t \right)^+ +\left(X^{m_x,\hat{g}}_{t} - D^{m_x}_t \right)^+
\nonumber\\
= & X^{M_x,\hat{g}}_{t} - D^{M_x}_t  +X^{m_x,\hat{g}}_{t} - D^{m_x}_t 
\nonumber\\
= & X^{M_x,\hat{g}}_{t}   +X^{m_x,\hat{g}}_{t}- \tilde{D}^{M_y}_t - \tilde{D}^{m_y}_t
\nonumber\\
\leq & Y^{M_y}_{t}   +Y^{m_y}_{t}- \tilde{D}^{M_y}_t - \tilde{D}^{m_y}_t
\nonumber\\
\leq &\left(Y^{M_y}_{t} - \tilde{D}^{M_y}_t \right)^++\left(Y^{m_y}_{t} - \tilde{D}^{m_y}_t \right)^+
\label{eq:pospos}
\end{align}
where the second equality in \eqref{eq:pospos} follows from \eqref{eq:DYX} and
the first inequality in \eqref{eq:pospos} follows from the induction hypothesis for \eqref{eq:stosumbdd}.
\\
The above results, namely \eqref{eq:YeqXsumeq}-\eqref{eq:pospos}, show that the right hand side of (\ref{eq:sumYXbar}) is positive, and the proof for (\ref{eq:xbarysum}) is complete.
\\
It remains to show that (\ref{eq:xbarymax}) is true.
\\
We first consider case 2.\\
In case 2, we know from (\ref{eq:ADtilde}) that
\begin{align}
&\left(\tilde{A}^{M_y}_t,\tilde{D}^{M_y}_t,\tilde{A}^{m_y}_t,\tilde{D}^{m_y}_t\right)= 
\left(A^{M_x}_t,D^{M_x}_t,A^{m_x}_t,D^{m_x}_t\right).
\label{eq:ADADeqADAD}
\end{align}
Then, 
\begin{align}
\overline{X}^{M_x,\hat{g}}_{t}
= &\left(X^{M_x,\hat{g}}_{t} - D^{M_x}_t \right)^+ + A^{M_x}_t  \nonumber\\
= &\left(X^{M_x,\hat{g}}_{t} - \tilde{D}^{M_y}_t \right)^+ +\tilde{A}^{M_y}_t \nonumber\\
\leq &\left(Y^{M_y}_{t} - \tilde{D}^{M_y}_t \right)^+ +\tilde{A}^{M_y}_t \nonumber\\
= & Y^{M_y}_{t+1},
\label{eq:MxleqMy}
\end{align}
where the second equality is a consequence of \eqref{eq:ADADeqADAD}
and the inequality follows from the induction hypothesis for \eqref{eq:stomaxbdd}.
\\
To proceed further we note that in case 2 there are three possibilities:
$\{Y^{M_y}_{t} = X^{M_x,\hat{g}}_{t} \}$,
$\{Y^{M_y}_{t}-2 \geq X^{m_x,\hat{g}}_{t} \}$ and 
$\{Y^{M_y}_{t} > X^{M_x,\hat{g}}_{t}, Y^{M_y}_{t}-2 < X^{m_x,\hat{g}}_{t}  \}$
\\
If $Y^{M_y}_{t} = X^{M_x,\hat{g}}_{t} $,
\eqref{eq:mymx} is also true. Following similar arguments as in (\ref{eq:MxleqMy}) we obtain
\begin{align}
\overline{X}^{m_x,\hat{g}}_{t} \leq Y^{m_y}_{t+1}.
\label{eq:poss1}
\end{align}
If $Y^{M_y}_{t}-2 \geq X^{m_x,\hat{g}}_{t} $
\begin{align}
\overline{X}^{m_x,\hat{g}}_{t} \leq X^{m_x,\hat{g}}_{t}+1 \leq Y^{M_y}_{t}-1 \leq Y^{M_y}_{t+1}.
\label{eq:poss2}
\end{align}
If $Y^{M_y}_{t} >X^{M_x,\hat{g}}_{t} $ and $Y^{M_y}_{t}-2 < X^{m_x,\hat{g}}_{t} $ it can only be
$Y^{M_y}_{t}-1 = X^{M_x,\hat{g}}_{t} = X^{m_x,\hat{g}}_{t} $. 
Since we are in case 2, $\left(A^{M_x}_t,D^{M_x}_t,A^{m_x}_t,D^{m_x}_t\right) \neq (0,1,1,0)$.
Therefore, 
\begin{align}
A^{m_x}_t-D^{m_x}_t \leq A^{M_x}_t-D^{M_x}_t+1.
\end{align}
Then we get
\begin{align}
\overline{X}^{m_x,\hat{g}}_{t} 
=& \left( Y^{M_y}_{t}-1 - D^{m_x}_t\right)^+ + A^{m_x}_t \nonumber\\
=& \max \left(A^{m_x}_t,
Y^{M_y}_{t}- 1 - D^{m_x}_t + A^{m_x}_t \right) \nonumber\\
\leq & \max \left(A^{m_x}_t,
 Y^{M_y}_{t}- D^{M_x}_t + A^{M_x}_t \right) \nonumber\\
\leq &\max \left(A^{m_x}_t, Y^{M_y}_{t+1} \right) \nonumber\\
\leq & \max \left(Y^{m_y}_{t+1}, Y^{M_y}_{t+1} \right).
\label{eq:poss3}
\end{align}
Combining \eqref{eq:MxleqMy}, \eqref{eq:poss1}, \eqref{eq:poss2} and \eqref{eq:poss3} we get
\eqref{eq:xbarymax} when case 2 is true.
\\
Now consider case 1. We have $Y^{M_y}_{t}-1 = X^{M_x,\hat{g}}_{t} = X^{m_x,\hat{g}}_{t} $.
\\
When $\left(A^{M_x}_t,D^{M_x}_t,A^{m_x}_t,D^{m_x}_t\right) = (0,1,1,0)$, then
\begin{align}
\overline{X}^{M_x,\hat{g}}_{t}
=& \left(X^{M_x,\hat{g}}_{t} -1 \right)^+\nonumber\\
\leq & \overline{X}^{m_x}_{t} \nonumber\\
=&  X^{m_x}_{t}+1 \nonumber\\      
=& \left(Y^{M_y}_{t}-D^{m_x}_t\right)^+ + A^{M_x}_t \nonumber\\     
= & Y^{M_y}_{t+1}
\label{eq:case1poss1}
\end{align}
When $\left(A^{M_x}_t,D^{M_x}_t,A^{m_x}_t,D^{m_x}_t\right) = (0,0,1,1)$ we get
\begin{align}
\overline{X}^{M_x,\hat{g}}_{t}
=& X^{M_x,\hat{g}}_{t}\nonumber\\
\leq & \overline{X}^{m_x,\hat{g}}_{t} \nonumber\\
=& \max \left(X^{m_x,\hat{g}}_{t}, 1 \right) \nonumber\\
= & \max \left( \left(Y^{M_y}_{t}-D^{m_x}_t\right)^+ + A^{M_x}_t, A^{m_x}_t \right) \nonumber\\
= & \max \left(Y^{M_y}_{t+1}, A^{m_x}_t \right) \nonumber\\
\leq & \max \left(Y^{M_y}_{t+1}, Y^{m_y}_{t+1} \right). 
\label{eq:case1poss2}
\end{align}
Combining \eqref{eq:case1poss1} and \eqref{eq:case1poss2} we obtain \eqref{eq:xbarymax} for case 1.
\\
As a result, \eqref{eq:xbarymax} holds for both cases 1 and 2.

\noindent\textit{Remark:}

We note that we need the two cases described in \eqref{eq:ADtilde} for the following reasons.
If we eliminate case 1 and always associate 
$\left(\tilde{A}^{M_y}_t,\tilde{D}^{M_y}_t,\tilde{A}^{m_y}_t,\tilde{D}^{m_y}_t\right)$
with $\left(A^{M_x}_t,D^{M_x}_t,A^{m_x}_t,D^{m_x}_t\right)$ as in case 2, then when
$\{Y^{M_y}_{t}-1 =X^{m_x,\hat{g}}_{t}$ and $\left(A^{M_x}_t,D^{M_x}_t,A^{m_x}_t,D^{m_x}_t\right) = (0,1,1,0)\}$,
the shorter queue $m_x$ increases by one customer, and the longer queue $M_y$ decreases by one customer; therefore
$\overline{X}^{m_x,\hat{g}}_{t} = Y^{M_y}_{t+1}+1$ and \eqref{eq:xbarymax} is not satisfied.

\end{proof}

\begin{proof}[Proof of Lemma \ref{lm:inf:bddunctrl}]
From Lemma \ref{lm:inf:stobdd}, at any time $t$ there exists $Y^{i}_t$ such that such that 
(\ref{eq:stoeq})-(\ref{eq:stomaxbdd}) hold.
\\
Adopting the notations $M_x,m_x$ and $M_y,m_y$ in the proof of Lemma \ref{lm:inf:stobdd}, 
we have at every time $t$
\begin{align}
&X^{m_x,\hat{g}}_t \leq X^{M_X,\hat{g}}_t \quad a.s.,\\
&Y^{m_y}_t \leq Y^{M_y}_t \quad a.s. 
\end{align}
Furthermore, from \eqref{eq:stomaxbdd} we have
\begin{align}
&X^{M_x,\hat{g}}_t \leq Y^{M_y}_t \quad a.s.
\label{eq:XMYMineq}
\end{align}
If $X^{m_x,\hat{g}}_t \leq Y^{m_y}_t$, \eqref{eq:XMYMineq} and the fact that $c(\cdot)$ is increasing give
\begin{align}
c\left(X^{M_X,\hat{g}}_t\right)+c\left(X^{m_X,\hat{g}}_t\right) \leq 
c\left(Y^{M_y}_t\right)+c\left(Y^{m_y}_t\right).
\end{align}
If $X^{m_x,\hat{g}}_t > Y^{m_y}_t$, then
\begin{align}
&Y^{m_y}_t < X^{m_x,\hat{g}}_t \leq X^{M_x,\hat{g}}_t  \leq Y^{M_y}_t . \label{pflm:inf:bddunctrl:eq1}
\end{align}
Since $c(\cdot)$ is convex, it follows from (\ref{pflm:inf:bddunctrl:eq1}) that
\begin{align}
\frac{c\left(Y^{M_y}_t\right)-c\left(X^{M_x,\hat{g}}_t\right)}{Y^{M_y}_t-X^{M_x,\hat{g}}_t} \geq 
\frac{c\left(X^{m_x,\hat{g}}_t\right) -c\left(Y^{m_y}_t\right)}{X^{m_x,\hat{g}}_t-Y^{m_y}_t} .
\label{pflm:inf:bddunctrl:eq2}
\end{align}
From (\ref{eq:stosumbdd}) in Lemma \ref{lm:inf:stobdd} we know that
\begin{align}
&Y^{M_y}_t - X^{M_x,\hat{g}}_t  \geq  X^{m_x,\hat{g}}_t - Y^{m_y}_t. \label{pflm:inf:bddunctrl:eq3}
\end{align}
Combining (\ref{pflm:inf:bddunctrl:eq2}) and (\ref{pflm:inf:bddunctrl:eq3}) we get
\begin{align}
c\left(Y^{M_y}_t\right)+c\left(Y^{m_y}_t\right) \geq
c\left(X^{M_x,\hat{g}}_t\right)+c\left(X^{m_x,\hat{g}}_t\right) .
\end{align}
\end{proof}

\begin{proof}[Proof of Lemma \ref{lm:inf:WTconv}]
Let $\{Y^{1}_t, t \in \mathbb{Z}_+ \}$ and $\{Y^{2}_t, t\in \mathbb{Z}_+\}$ be the processes defined in Lemma \ref{lm:inf:stobdd}.
Then $\{Y^{i}_t, t\in \mathbb{Z}_+\}$ has the same distribution as $\{X^{i,g_0}_t, t\in \mathbb{Z}_+\}$ for $i=1,2$.
\\
Since $\mu>\lambda$, the processes $\{Y^{i}_t, t\in \mathbb{Z}_+\},i=1,2$ are irreducible positive recurrent Markov chains.
Moreover, the two processes $\{Y^{1}_t, t \in \mathbb{Z}_+ \}$ and $\{Y^{2}_t, t\in \mathbb{Z}_+\}$ have the same stationary distribution, denoted by $\pi^{g_0}$.
Under Assumption \ref{assum:initial}, by Ergodic theorem of Markov chains (see \cite[chap. 3]{bremaud1999markov}) we get 
\begin{align}
\lim_{T \rightarrow \infty} \frac{1}{T} \sum_{t=0}^{T-1}c(Y^{1}_t) 
= &\lim_{T \rightarrow \infty} \frac{1}{T}\sum_{t=0}^{T-1}c(Y^{2}_t) \nonumber\\
= &\sum_{x=0}^{\infty} \pi^{g_0}(x) c(x)\quad a.s.
\end{align}
Let $W^i_T(Y_{0:T-1}) := \frac{1}{T}\sum_{t=0}^{T-1}c(Y^{i}_t), i=1,2$.\\
We show that $\{W^i_T(Y_{0:T-1}) , T=1,2,\dots\}$ is uniformly integrable for $i=1,2$.
That is, 
\begin{align}
\sup_{T} \mathbf{E}\left[ W^i_T(Y_{0:T-1}) 1_{\{W^i_T(Y_{0:T-1}) >N\}} \right] \rightarrow 0
\end{align}
as $N\rightarrow\infty $.
\\
Let $p^{g_0}(x,y),x,y\in \mathbb{Z}_+$ be the transition probabilities of the Markov chain.
Note that the initial PMF of the process $\{Y^{i}_t, t\in \mathbb{Z}_+\},i=1,2$ is $\pi^i_0$.
From Assumption \ref{assum:initial} we know that $\pi^i_0(x)=0,i=1,2$ for all $x>M$.
\\
Letting $R := \max_{x\leq M} \frac{\pi^i_0(x)}{\pi^{g_0}(x)}<\infty$, we obtain for $i=1,2$
\begin{align}
  &\mathbf{E}\left[ W^i_T(Y_{0:T-1}) 1_{\{W^i_T(Y_{0:T-1}) >N\}} \right]\nonumber\\
= &\sum_{y_{0:T-1}} W^i_T(y_{0:T-1})1_{\{W^i_T(y_{0:T-1}) >N\}}\mathbf{P}(Y_{0:T-1} = y_{0:T-1}) \nonumber\\
= &\sum_{y_{0:T-1}} W^i_T(y_{0:T-1})1_{\{W^i_T(y_{0:T-1}) >N\}}
\pi^{i}_0(y_0)\Pi_{t=1}^{T-1}p^{g_0}(y_{t-1},y_t)
 \nonumber\\
\leq &R\sum_{y_{0:T-1}} W^i_T(y_{0:T-1})1_{\{W^i_T(y_{0:T-1}) >N\}}
\pi^{g_0}(y_0)\Pi_{t=1}^{T-1}p^{g_0}(y_{t-1},y_t)
 \nonumber\\
= & R \mathbf{E}\left[ W^{\pi^{g_0}}_T 1_{\{W^{\pi^{g_0}}_T >N\}} \right],
\label{eq:WUI}
\end{align}
where $W^{\pi^{g_0}}_T = \frac{1}{T}\sum_{t=0}^{T-1}c(Y^{\pi^{g_0}}_t)$ and 
$\{Y^{\pi^{g_0}}_t, t\in \mathbb{Z}_+\}$ is the chain with transition probabilities $p^{g_0}(x,y)$ and initial PMF $\pi^{g_0}$.
\\
Note that $\{Y^{\pi^{g_0}}_t, t\in \mathbb{Z}_+\}$ is stationary because the initial PMF is the stationary distribution $\pi^{g_0}$.
From Birkhoff's Ergodic theorem we know that $\{W^{\pi^{g_0}}_T,T=1,2,\dots\}$ converges $a.s.$ and in expectation (see \cite[chap. 2]{petersen1989ergodic}). 
Therefore, $\{W^{\pi^{g_0}}_T,T=1,2,\dots\}$ is uniformly integrable, and the right hand side of (\ref{eq:WUI}) goes to zeros uniformly as $N \rightarrow \infty$.
Consequently, $\{W^i_T(Y_{0:T-1}) , T=1,2,\dots\}$ is also uniformly integrable for $i=1,2$.
\\
Since $W_T = W^1_T(Y_{0:T-1})+W^2_T(Y_{0:T-1})$ for all $T=1,2,\dots$, $\{W_T,T=1,2,\dots\}$	is uniformly integrable.

\end{proof}
\begin{proof}[Proof of Corollary \ref{cor:infcostfinite}]
From Lemma \ref{lm:inf:bddunctrl}, there exists  $\{Y^{1}_t,Y^2_t, t \in \mathbb{Z}_+\}$ such that (\ref{eq:stoeq}) holds and
\begin{align}
c\left(X^{1,\hat{g}}_t\right)+c\left(X^{2,\hat{g}}_t\right) \leq 
c\left(Y^{1}_t\right)+c\left(Y^{2}_t\right) \quad a.s.
\label{eq:XlessY}
\end{align}
Let
\begin{align}
W_T :=& \frac{1}{T}\sum_{t=0}^{T-1}
\left(c\left(Y^1_t\right)+c\left(Y^2_t\right)\right),\\
V_T :=& \frac{1}{T}\sum_{t=0}^{T-1}
\left(c\left(X^{1,\hat{g}}_t\right)+c\left(X^{2,\hat{g}}_t\right)\right).
\end{align}
From (\ref{eq:XlessY}) it follows that
\begin{align}
V_T \leq W_T, T=1,2,\dots
\end{align}
From Lemmas \ref{lm:inf:WTconv}, $\{W_T, T =1,2,\dots\} $ is uniformly integrable, therefore $\{V_T, T =1,2,\dots\} $, which is bounded above by $\{W_T, T =1,2,\dots\} $  is also uniformly integrable.
\\
From the property of uniformly integrability, if $\{V_T, T =1,2,\dots\} $ converges a.s., we know that $\{V_T, T =1,2,\dots\} $ also converges in expectation. Furthermore,
\begin{align}
J^{\hat{g}}\left(\pi^1_0,\pi^2_0\right)
= & \limsup_{T\rightarrow\infty}\frac{1}{T}
\mathbf{E}\left[
\sum_{t=0}^{T-1}
\left(c\left(Y^1_t\right)+c\left(Y^2_t\right)\right)\right]\nonumber\\
=& \limsup_{T\rightarrow\infty}
\mathbf{E}\left[V_T\right]\nonumber\\
\leq & \limsup_{T\rightarrow\infty}
\mathbf{E}\left[W_T\right] = J^{g_0}.
\end{align}
\end{proof}

\section{Proofs of the Results Associated with Step 2 of the Proof of Theorem \ref{thm:inf}}
\label{app:infinite2}

\begin{proof}[Proof of Lemma \ref{lm:inf:MC}]
First we show that $\{S_t,t\geq T_0+1\}$ is a Markov chain.\\
For $s_t \geq 2$, 
\begin{align}
&\mathbf{P}\left( S_{t+1}=s_{t+1}|S_{T_0+1:t}=s_{T_0+1:t}\right) \nonumber\\
= &\mathbf{P}\left( \left(s_t-D^1_t -D^2_t +A^1_t+A^2_t \right)=s_{t+1}\right.\nonumber\\
&\qquad \qquad\qquad \left.|S_{T_0+1:t}=s_{T_0+1:t}\right)\nonumber\\
= &\mathbf{P}\left( \left( s_t-D^1_t -D^2_t +A^1_t+A^2_t\right)=s_{t+1}|S_t=s_t\right)\nonumber\\
=&\mathbf{P}\left( S_{t+1}=s_{t+1}|S_t=s_t\right).
\label{pflm:inf:MC:eq1}
\end{align}
The first and last equalities in \eqref{pflm:inf:MC:eq1} follow from the construction of the process $\{S_t, t\geq T_0+1\}$.
The second equality in \eqref{pflm:inf:MC:eq1} is true because 
$T_0$ is a stopping time
 with respect to $\{X^{1,\hat{g}}_t,X^{2,\hat{g}}_t, t\in \mathbb{Z}_+\}$,
 and $A^i_t, D^i_t, i=1,2$ are independent of all random variables before $t$.
Similarly, for $s_t =0$ we have, by arguments similar to the above,
\begin{align}
&\mathbf{P}\left( S_{t+1}=s_{t+1}|S_{T_0+1:t}=s_{T_0+1:t}\right) \nonumber\\
= &\mathbf{P}\left( A^1_t+A^2_t=s_{t+1}|S_{T_0+1:t-1}=s_{T_0+1:t-1}, S_t=0\right)\nonumber\\
= &\mathbf{P}\left( A^1_t+A^2_t=s_{t+1}|S_t=0\right)\nonumber\\
=&\mathbf{P}\left( S_{t+1}=s_{t+1}|S_t=0\right).
\label{pflm:inf:MC:eq2}
\end{align}
The first and last equality in \eqref{pflm:inf:MC:eq2} follow from the construction of the process $\{S_t, t\geq T_0+1\}$.
The second equality in \eqref{pflm:inf:MC:eq2} is true because $A^i_t, D^i_t, i=1,2$ are independent of all variables before $t$.
For $s_t =1$, 
\begin{align}
&\mathbf{P}\left( S_{t+1}=s_{t+1}|S_{T_0+1:t}=s_{T_0+1:t}\right) \nonumber\\
= &\mathbf{P}\left( s_t+1_{\left\{X^{1,\hat{g}}_t=0\right\}}(D^1_t- D^2_t)-D^1_t +A^1_t+A^2_t
 =s_{t+1} |S_{T_0+1:t}=s_{T_0+1:t}\right)
\nonumber\\
= &\mathbf{P}\left( 1 -D^2_t +A^1_t+A^2_t=s_{t+1}, 
X^{1,\hat{g}}_t=0|S_{T_0+1:t-1}=s_{T_0+1:t-1}, S_t=1\right)\nonumber\\
&+ \mathbf{P}\left( 1 -D^1_t +A^1_t+A^2_t=s_{t+1},
 X^{1,\hat{g}}_t=1|S_{T_0+1:t-1}=s_{T_0+1:t-1}, S_t=1\right)
\nonumber\\
= &\mathbf{P}\left( 1 -D^1_t +A^1_t+A^2_t=s_{t+1}, 
  X^{1,\hat{g}}_t=0|S_{T_0+1:t-1}=s_{T_0+1:t-1}, S_t=1\right)\nonumber\\
&+ \mathbf{P}\left( 1 -D^1_t +A^1_t+A^2_t=s_{t+1},
  X^{1,\hat{g}}_t=1|S_{T_0+1:t-1}=s_{T_0+1:t-1}, S_t=1\right)
\nonumber\\
= &\mathbf{P}\left( 1 -D^1_t +A^1_t+A^2_t=s_{t+1} 
|S_{T_0+1:t-1}=s_{T_0+1:t-1}, S_t=1\right)
\nonumber\\
= &\mathbf{P}\left( 1 -D^1_t +A^1_t+A^2_t=s_{t+1}
|S_t=1\right)\nonumber\\
=&\mathbf{P}\left( S_{t+1}=s_{t+1}|S_t=s_t\right).
\label{pflm:inf:MC:eq3}
\end{align}
The first equality in \eqref{pflm:inf:MC:eq3} follows from the construction of the process $\{S_t, t\geq T_0+1\}$.
The second and forth equalities follow from the fact that $X^{1,\hat{g}}_t$ can be either $0$ or $1$.
In the third equality, $D^2_t$ is replaced by $D^1_t$ in the first term;
this is true 
because $D^1_t$ and $D^2_t$ are identically distributed and independent of $X^{1,\hat{g}}_t$ and all past random variables.
The fifth equality holds because $T_0$ is a stopping time
 with respect to $\{X^{1,\hat{g}}_t,X^{2,\hat{g}}_t, t\in \mathbb{Z}_+\}$ and $A^i_t, D^i_t, i=1,2$ are independent of all past random variables.
The last equality follows from the same arguments that lead to the first through the fifth equalities.
\\\\
Therefore, the process $\{S_t, t\geq T_0+1\}$ is a Markov chain.
\\
Since $\lambda,\mu >0$, the Markov chain is irreducible.
\\
We prove that the process $\{S_t, t\geq T_0+1\}$ is positive recurrent.
Note that, for all $s=0,1,2,\dots$, because of the construction of $\{S_t, t\geq T_0+1\}$ 
\begin{align}
&\mathbf{E}\left[ S_{t+1} | S_t=s\right] \nonumber\\
\leq & \mathbf{E}\left[ S_t +A^1_t+A^2_t| S_t=s\right]\nonumber\\
= & s +2\lambda < \infty.
\end{align}
Moreover, for all $s\geq 2$, 
\begin{align}
&\mathbf{E}\left[ S_{t+1} | S_t=s\right] \nonumber\\
= & \mathbf{E}\left[ s-D^1_t -D^2_t +A^1_t+A^2_t | S_t=s\right]\nonumber\\
= & s-2\mu +2\lambda < s.
\end{align}
Using Foster's theorem (see \cite[chap. 5]{bremaud1999markov}), we conclude that the Markov chain $\{S_t, t\geq T_0+1\}$ is positive recurrent.
\end{proof}

\begin{proof}[Proof of Lemma \ref{lm:inf:00io}]
Let $(\Omega, \mathcal{F}, \mathbf{P})$ denote the basic probability space for our problem.
Define events $E_t \in \mathcal{F}, t=0,1,\dots $ to be
\begin{align}
E_t =& \{\omega \in \Omega : \left(U^{1,\hat{g}}_{t'}(\omega),U^{2,\hat{g}}_{t'}(\omega)\right)\neq(0,0) \quad \forall t' \geq t\}
\end{align}
If the claim of this lemma is not true, we get
\begin{align}
\mathbf{P}\left(\bigcup_{t =0}^{\infty}E_t\right)=
1-\mathbf{P}\left( \left(U^{1,\hat{g}}_{t},U^{2,\hat{g}}_{t}\right)=(0,0) \quad i.o.\right)>0.
\end{align}
Therefore, there exist some $t_0$ such that $\mathbf{P}(E_{t_0}) >0$.
Since $t_0$ is a constant, it is a stopping time with respect to $\{X^{1,\hat{g}}_t,X^{2,\hat{g}}_t, t \in \mathbb{Z}_+ \}$.
\\
Consider the process $\{S_t,t=t_0+1,t_0+2,...\}$ defined in Lemma \ref{lm:inf:MC} with the stopping time $t_0$. 
From Lemma \ref{lm:inf:MC} we know that $\{S_t, t\geq t_0+1 \}$ is an irreducible positive recurrent Markov chain.
Furthermore, along the sample path induced by any $\omega \in E_{t_0}$, we claim that for all $t \geq t_0+1$
\begin{align}
S_t(\omega)
=&X^{1,\hat{g}}_{t}(\omega)+X^{2,\hat{g}}_{t}(\omega) \nonumber\\
= &\overline{X}^{1,\hat{g}}_{t-1}(\omega)+\overline{X}^{2,\hat{g}}_{t-1}(\omega).
\label{pflm:inf:00io:eqsum}
\end{align}
The claim is shown by induction below.
\\
By the definition of $\{S_t, t\geq t_0+1 \}$ in Lemma \ref{lm:inf:MC}, we have at time $t_0+1$ for any $\omega \in E_{t_0}$
\begin{align}
S_{t_0+1}(\omega) = &X^{1,\hat{g}}_{t_0+1}(\omega)+X^{2,\hat{g}}_{t_0+1}(\omega) \nonumber\\
			      = &\overline{X}^{1,\hat{g}}_{t_0}(\omega)+\overline{X}^{2,\hat{g}}_{t_0}(\omega),
\label{pflm:inf:00io:eqbas}
\end{align}
where the last inequality in \eqref{pflm:inf:00io:eqbas} follows from the system dynamics \eqref{Model:dynamic1}-\eqref{Model:dynamic3}.
\\
Assume equation \eqref{pflm:inf:00io:eqsum} is true at time $t$ ($t\geq t_0+1$).
At time $t+1$ we have, by \eqref{Model:dynamic1}-\eqref{Model:dynamic3},
\begin{align}
&X^{1,\hat{g}}_{t+1}+X^{2,\hat{g}}_{t+1}\nonumber\\
=& (X^{1,\hat{g}}_t-D^1_t)^++(X^{2,\hat{g}}_t-D^2_t)^+ +A^1_t+A^2_t\nonumber\\
=& X^{1,\hat{g}}_t+X^{2,\hat{g}}_t-D^1_t-D^2_t +A^1_t+A^2_t\nonumber\\
&+D^1_t1_{\left\{X^{1,\hat{g}}_t=0\right\}}+D^2_t1_{\left\{X^{2,\hat{g}}_t=0\right\}}.
\end{align}
Since along the sample path induced by $\omega \in E_{t_0}$, $\left(U^{1,\hat{g}}_{t-1}(\omega),U^{2,\hat{g}}_{t-1}(\omega)\right)\neq (0,0)$ and $X^{i,\hat{g}}_t=\overline{X}^{i,\hat{g}}_{t-1}-U^{i,\hat{g}}_{t-1}+U^{j,\hat{g}}_{t-1}$, 
the event 
$\{X^{i,\hat{g}}_t=0\} \bigcap E_{t_0}$ ($i=1$ or $2$) implies that 
$\overline{X}^{i,\hat{g}}_{t-1}=1,\, U^{i,\hat{g}}_{t-1}=1 \text{ and  }U^{j,\hat{g}}_{t-1}=0$.
For this case, $\overline{X}^{i,\hat{g}}_{t-1}=1$ and $U^{i,\hat{g}}_{t-1}=1$ further imply that the threshold is smaller than one. Then, the only possibility for $U^{j,\hat{g}}_{t-1}=0$ is $\overline{X}^{j,\hat{g}}_{t-1}=0$. 
Therefore,
\begin{align}
&\left\{X^{i,\hat{g}}_t=0\right\} \bigcap E_{t_0} \nonumber\\
\subseteq 
&\left\{\overline{X}^{i,\hat{g}}_{t-1}=1,\, U^{i,\hat{g}}_{t-1}=1, \, \overline{X}^{j,\hat{g}}_{t-1}=0 \text{ and  }U^{j,\hat{g}}_{t-1}=0\right\}\nonumber\\
\subseteq &
\{S_t = 1\}.
\label{eq:Seq1}
\end{align}
Consequently, from \eqref{eq:Seq1}, for any $\omega \in E_{t_0}$
\begin{align}
&D^1_t(\omega)1_{\left\{X^{1,\hat{g}}_t(\omega)=0\right\}}+D^2_t(\omega)1_{\left\{X^{2,\hat{g}}_t(\omega)=0\right\}}\nonumber\\
=& 1_{\left\{S_t(\omega)=1\right\}}
\left(D^1_t(\omega)1_{\left\{X^{1,\hat{g}}_t(\omega)=0\right\}}
+ D^2_t(\omega)1_{\left\{X^{2,\hat{g}}_t(\omega)=0\right\}} \right).\nonumber\\
=& 1_{\left\{S_t(\omega)=1\right\}}\left(1_{\left\{X^{1,\hat{g}}_t(\omega)=0\right\}}(D^1_t(\omega)-D^2_t(\omega))+ D^2_t(\omega) \right).
\label{eq:forS1}
\end{align}
Moreover, $\left(U^{1,\hat{g}}_{t-1}(\omega),U^{2,\hat{g}}_{t-1}(\omega)\right)\neq (0,0)$ implies that $\left(\overline{X}^{1,\hat{g}}_{t-1}(\omega),\overline{X}^{2,\hat{g}}_{t-1}(\omega)\right)\neq (0,0)$.
Hence, 
\begin{align}
S_t(\omega) = \overline{X}^{1,\hat{g}}_{t-1}(\omega)+\overline{X}^{2,\hat{g}}_{t-1}(\omega)\neq 0,
\label{eq:forS0}
\end{align}
and 
\begin{align}
&X^{1,\hat{g}}_{t+1}(\omega)+X^{2,\hat{g}}_{t+1}(\omega)\nonumber\\
=& X^{1,\hat{g}}_t(\omega)+X^{2,\hat{g}}_t(\omega)-D^1_t(\omega)-D^2_t(\omega) +A^1_t(\omega)+A^2_t(\omega)\nonumber\\
&+ 1_{\left\{S_t(\omega)=1\right\}}
\left(1_{\left\{X^{1,\hat{g}}_t(\omega)=0\right\}}(D^1_t(\omega)- D^2_t(\omega))+D^2_t(\omega) \right)\nonumber\\
=& X^{1,\hat{g}}_t(\omega)+X^{2,\hat{g}}_t(\omega)-D^1_t(\omega)-D^2_t(\omega) +A^1_t(\omega)+A^2_t(\omega)\nonumber\\
&+ 1_{\left\{S_t(\omega)=1\right\}}
\left(1_{\left\{X^{1,\hat{g}}_t(\omega)=0\right\}}(D^1_t(\omega)- D^2_t(\omega))+D^2_t(\omega) \right)\nonumber\\
&+1_{\{S_t(\omega)=0\}}\left(D^1_t(\omega)+D^2_t(\omega) \right)\nonumber\\
=& S_{t+1}(\omega),
\label{eq:Sinduction}
\end{align}
where the first and second equalities in \eqref{eq:Sinduction} follow from \eqref{eq:forS1} and \eqref{eq:forS0}, respectively.
The last equality in \eqref{eq:Sinduction} follows from the construction of $\{S_t, t\geq t_0+1\}$.
\\
Furthermore, by the system dynamics (\ref{Model:dynamic1})-(\ref{Model:dynamic3}) we have 
\begin{align}
\overline{X}^{1,\hat{g}}_{t}(\omega)+\overline{X}^{2,\hat{g}}_{t}(\omega)
= &X^{1,\hat{g}}_{t+1}(\omega)+X^{2,\hat{g}}_{t+1}(\omega)\nonumber\\
=& S_{t+1}(\omega).
\end{align}
Thus, equation (\ref{pflm:inf:00io:eqsum}) is true for any $\omega \in E_{t_0}$ for all $t \geq t_0+1$.
\\
Then, for any $\omega \in E_{t_0}$ 
\begin{align}
S_t(\omega) = \overline{X}^{1,\hat{g}}_{t-1}(\omega)+\overline{X}^{2,\hat{g}}_{t-1}(\omega) \neq 0 \text{ for all }t \geq t_0+1
\label{eq:Scontradict}
\end{align}
because $\left(U^{1,\hat{g}}_{t-1}(\omega),U^{2,\hat{g}}_{t-1}(\omega)\right)\neq (0,0)$ for all $t \geq t_0+1$.
Since $\mathbf{P}(E_{t_0})>0 $, \eqref{eq:Scontradict} contradicts the fact that $\{S_t, t\geq t_0+1\}$ is recurrent.
\\
Therefore, no such event $E_{t_0}\in \mathcal{F}$ with positive probability exists, and the proof of this lemma is complete.

\end{proof}

\section{Proofs of the Results Associated with Step 3 of the Proof of Theorem \ref{thm:inf}}
\label{app:lmcen}

\begin{proof}[Proof of Lemma \ref{lm:cen}]
For any fixed centralized policy $g\in \mathcal{G}_c$, the information $I^1_t,I^2_t$ available to the centralized controller includes all primitive random variables $X^i_{0},A^i_{0:t},D^i_{0:t},i=1,2$
up to time $t$.
Since all other random variables are functions of these primitive random variables and $g$, 
we have
\begin{align}
U^{i,g}_t=& g^i_t(I^1_t,I^2_t)\nonumber\\
	     =& g^i_t(X^1_{0},X^2_{0},A^1_{0:t},A^2_{0:t},D^1_{0:t},D^2_{0:t}),
\end{align}
for $i=1,2$.
For any initial queue lengths $x^1_{0},x^2_{0}$, we now define a policy $\tilde{g}$ from $g$ for the case when both queues are initially empty. 
Let $\tilde{g}$ be the policy such that for $i=1,2$
\begin{align}
U^{i,\tilde{g}}_t = &\tilde{g}^i_t(I^1_t,I^2_t) \nonumber\\
:=&
\left\lbrace
\begin{array}{ll}
g^i_t(x^1_{0},x^2_{0},A^1_{0:t},A^2_{0:t},D^i_{0:t},D^2_{0:t}) & \text{ if }\overline{X}^{i,\tilde{g}}_t>0\\
0 & \text{ if }\overline{X}^{i,\tilde{g}}_t=0
\end{array}
\right. \nonumber\\
=& \min \left( U^{i,g}_t, \overline{X}^{i,\tilde{g}}_t \right) \leq U^{i,g}_t,
\label{pflmcen:UU}
\end{align}
where $X^{1,\tilde{g}}_t$ and $X^{2,\tilde{g}}_t$ denote the queue lengths at time $t$ due to policy $\tilde{g}$ with initial queue lengths 
$X^{1,\tilde{g}}_0=X^{2,\tilde{g}}_0=0$.
\\
At time $0$ we have $X^{i,g}_0=x^i_0\geq 0 = X^{i,\tilde{g}}_0$ for $i=1,2$.
We now prove by induction that for all time $t$
\begin{align}
X^{i,g}_t \geq X^{i,\tilde{g}}_t, \quad i=1,2.
\label{pflmcen:claim}
\end{align}
Suppose the claim is true at time $t$.
Then, from the system dynamics \eqref{Model:dynamic1}-\eqref{Model:dynamic2} and \eqref{pflmcen:claim} we obtain, for $i=1,2$,
\begin{align}
\overline{X}^{i,g}_t =& \left(X^{i,g}_t  - D^{i}_t \right)^+ +A^i_t \nonumber\\
                    \geq &\left(X^{i,\tilde{g}}_t  - D^{i}_t \right)^+ +A^i_t
                     = \overline{X}^{i,\tilde{g}}_t.
\label{pflmcen:eq1}
\end{align}
Furthermore from \eqref{Model:dynamic1}-\eqref{Model:dynamic2} and \eqref{pflmcen:UU}
\begin{align}
X^{i,g}_{t+1} = & \overline{X}^{i,g}_t-U^{i,g}_t+U^{j,g}_t \nonumber\\
             \geq &\overline{X}^{i,g}_t-U^{i,g}_t+U^{j,\tilde{g}}_t
             \label{pflmcen:eq2}
\end{align}
If $\overline{X}^{i,\tilde{g}}_t>0 $, then, because of \eqref{pflmcen:UU} and \eqref{pflmcen:eq1}
\begin{align}
\overline{X}^{i,g}_t-U^{i,g}_t =&\overline{X}^{i,g}_t-\min \left( U^{i,g}_t, \overline{X}^{i,\tilde{g}}_t \right)
\nonumber\\
= &\overline{X}^{i,g}_t-U^{i,\tilde{g}}_t
\geq \overline{X}^{i,\tilde{g}}_t-U^{i,\tilde{g}}_t
.
\label{pflmcen:eq3}
\end{align}
If $\overline{X}^{i,\tilde{g}}_t=0 $, since $\overline{X}^{i,g}_t-U^{i,g}_t \geq 0$, \eqref{pflmcen:UU} implies
\begin{align}
\overline{X}^{i,g}_t-U^{i,g}_t
\geq  0
= \overline{X}^{i,\tilde{g}}_t-U^{i,\tilde{g}}_t.
\label{pflmcen:eq4}
\end{align}
Combining \eqref{pflmcen:eq2}-\eqref{pflmcen:eq4} and \eqref{Model:dynamic1}-\eqref{Model:dynamic2} we get
\begin{align}
X^{i,g}_{t+1} \geq &\overline{X}^{i,g}_t-U^{i,g}_t+U^{j,\tilde{g}}_t \nonumber\\
\geq & \overline{X}^{i,\tilde{g}}_t-U^{i,\tilde{g}}_t+U^{j,\tilde{g}}_t
= X^{i,\tilde{g}}_{t+1}.
\end{align}
Therefore, we complete the proof of the claim \eqref{pflmcen:claim}.
\\
Since the cost function is increasing, (\ref{pflmcen:claim}) implies that for all $g\in \mathcal{G}_c$ and any initial condition
$X^1_0=x^1_0,X^2_0=x^2_0$,
\begin{align}
\inf_{g\in \mathcal{G}_c} J^g_T(0,0) \leq J^{\tilde{g}}_T(0,0) \leq J^{g}_T(x^1_0,x^2_0).
\end{align}
Consequently, for any PMFs $\pi^1_0,\pi^2_0$
\begin{align}
\inf_{g\in \mathcal{G}_c} J^g_T(0,0) \leq \inf_{g\in \mathcal{G}_c} J^{g}_T(\pi^1_0,\pi^2_0).
\end{align}
Moreover, the result of Lemma \ref{lm:comp} ensures that $\hat{g}$ gives the smallest expected cost among policies in $\mathcal{G}_c$ for any finite horizon when $X^1_0=X^2_0=0$. It follows that, for any finite $T$,
\begin{align}
J^{\hat{g}}_T(0,0) = \inf_{g\in \mathcal{G}_c} J^g_T(0,0) \leq J^{\tilde{g}}_T(0,0) \leq J^{g}_T(x^1_0,x^2_0).
\label{eq:finitecostorder}
\end{align}
For infinite horizon cost, we divide each term in (\ref{eq:finitecostorder}) by $T$ and let $T$ to infinity, and we obtain, for any $\pi^1_0,\pi^2_0$,
\begin{align}
J^{\hat{g}}(0,0) = \inf_{g\in \mathcal{G}_c} J^g(0,0) \leq J^{\tilde{g}}(0,0) \leq J^{g}(x^1_0,x^2_0).
\end{align}

\end{proof}

\section{Proofs of the Results Associated with Step 4 of the Proof of Theorem \ref{thm:inf}}
\label{app:claiminfthminf}

\begin{proof}[Proof of the claim in the proof of Theorem \ref{thm:inf}]

We prove here our claim expressed by equation (\ref{eq:pfinf:claim}) to complete the proof of Theorem \ref{thm:inf}.
By (\ref{pflm:inf:MC:eqS1}),
\begin{align}
S_{T_0+1} = X^{1,\hat{g}}_{T_0+1}+X^{2,\hat{g}}_{T_0+1}.
\end{align}
We prove by induction that $X^{1,\hat{g}}_{t}+X^{2,\hat{g}}_{t}= S_t$ for all $t \geq T_0+1$.
\\
Assume that $X^{1,\hat{g}}_{t}+X^{2,\hat{g}}_{t}= S_t$ at time $t$, $t\geq T_0+1$. 
Then for time $t+1$, because of the systems dynamics (\ref{Model:dynamic1})-(\ref{Model:dynamic3}),
\begin{align}
&X^{1,\hat{g}}_{t+1}+X^{2,\hat{g}}_{t+1}\nonumber\\
=& (X^{1,\hat{g}}_t-D^1_t)^++(X^{2,\hat{g}}_t-D^2_t)^+ +A^1_t+A^2_t\nonumber\\
=& X^{1,\hat{g}}_t+X^{2,\hat{g}}_t-D^1_t-D^2_t +A^1_t+A^2_t\nonumber\\
&+D^1_t1_{\left\{X^{1,\hat{g}}_t=0\right\}}+D^2_t1_{\left\{X^{2,\hat{g}}_t=0\right\}}.
\label{eq:pfinf:ind1}
\end{align}
When $X^{i,\hat{g}}_t=0$ ($i=1$ or $2$), $U^{j,\hat{g}}_{t-1}$ should be $0$
because 
\begin{align}
0=X^{i,\hat{g}}_t=\overline{X}^{i,\hat{g}}_{t-1}-U^{i,\hat{g}}_{t-1}+U^{j,\hat{g}}_{t-1}
\label{eq:pfinf:eqdn}
\end{align}
and $\overline{X}^{i,\hat{g}}_{t-1}-U^{i,\hat{g}}_{t-1}\geq 0$.\\
We consider the following two cases separately:
\begin{enumerate}[{Case} 1]
\item $U^{i,\hat{g}}_{t-1}=0$.
\item $U^{i,\hat{g}}_{t-1}=1$.
\end{enumerate}
\begin{enumerate}[{Case} 1]
\item When $U^{i,\hat{g}}_{t-1}=0$, we must have $\overline{X}^{i,\hat{g}}_{t-1}=0$ by (\ref{eq:pfinf:eqdn}). 
Then $\overline{X}^{j,\hat{g}}_{t-1}\in \{0,1\}$ for the following reason. 
When $U^{i,\hat{g}}_{t-1}=U^{j,\hat{g}}_{t-1}=0$, the sizes of both queues are between the lower bound and the threshold.
That is
\begin{align}
\overline{LB}^{\hat{g}}_{t-1}
\leq &\overline{X}^{i,\hat{g}}_{t-1}
\leq \left\lceil TH_t
\right \rceil-1,
\label{eq:pfinf:eqbd1p1}
\\
\overline{LB}^{\hat{g}}_{t-1}
\leq &\overline{X}^{j,\hat{g}}_{t-1}  
\leq \left\lceil TH_t
\right \rceil-1.
\label{eq:pfinf:eqbd1}
\end{align}
Combining \eqref{eq:pfinf:eqbd1p1}, \eqref{eq:pfinf:eqbd1} with $\overline{X}^{i,\hat{g}}_{t-1}=0$ we obtain
\begin{align}
\overline{X}^{j,\hat{g}}_{t-1} =
& \left|\overline{X}^{j,\hat{g}}_{t-1} - \overline{X}^{i,\hat{g}}_{t-1} \right| \nonumber\\
\leq &\left\lceil
TH_t
\right \rceil-1
-\overline{LB}^{\hat{g}}_{t-1}\nonumber\\
\leq 
&\frac{1}{2}\left( 
\overline{UB}^{\hat{g}}_{t-1}-\overline{LB}^{\hat{g}}_{t-1} \right) \leq 1.5,
\label{eq:pfinf:eqbd2}
\end{align}
where the last inequality in (\ref{eq:pfinf:eqbd2}) is true because of (\ref{eq:UBLBbarupdate}) in Lemma \ref{lm:bounds},  \eqref{eq:pfinf:UBLBless1}, and 
\begin{align*}
\overline{UB}^{\hat{g}}_{t-1}-\overline{LB}^{\hat{g}}_{t-1} \leq
&UB^{\hat{g}}_{t}+1-LB^{\hat{g}}_{t}+1 \leq 3 .
\end{align*}
Therefore, $\overline{X}^{j,\hat{g}}_{t-1}\leq 1$ because $\overline{X}^{j,\hat{g}}_{t-1}$ takes integer values.
\item When $U^{i,\hat{g}}_{t-1}=1$, we must have $\overline{X}^{i,\hat{g}}_{t-1}=1$ by (\ref{eq:pfinf:eqdn}). 
This implies that the threshold is not more than $1$, and the only possible value of $\overline{X}^{j,\hat{g}}_{t-1}$ less than the threshold is $0$.
\end{enumerate}
As a consequence of the above analysis for the cases 1 and 2, $\{X^{i,\hat{g}}_t=0 \}$ implies 
\begin{align}
S_t = \overline{X}^{i,\hat{g}}_{t-1}+\overline{X}^{j,\hat{g}}_{t-1} \leq 1.
\end{align}
Thus, for $i=1,2$,
\begin{align}
\left\lbrace X^{i,\hat{g}}_t=0 \right\rbrace
= \left\lbrace X^{i,\hat{g}}_t=0, S_t \leq 1 \right\rbrace.
\end{align}
Then,
\begin{align}
&D^1_t1_{\left\{X^{1,\hat{g}}_t=0\right\}}+D^2_t1_{\left\{X^{2,\hat{g}}_t=0\right\}}\nonumber\\
=&D^1_t1_{\left\{X^{1,\hat{g}}_t =0, S_t \leq 1\right\}}+D^2_{t+1}1_{\left\{X^{2,\hat{g}}_t=0, S_t \leq 1\right\}}\nonumber\\
=&D^1_t1_{\left\{X^{1,\hat{g}}_t=0, S_t = 1\right\}}+D^2_t1_{\left\{X^{1,\hat{g}}_t\neq 0, S_t = 1\right\}}\nonumber\\
&+
D^1_t1_{\left\{ S_t =0\right\}}+D^2_t1_{\left\{S_t =0\right\}}.
\label{eq:pfinf:ind2}
\end{align}
Combining (\ref{eq:pfinf:ind1}) and (\ref{eq:pfinf:ind2}) we obtain
\begin{align}
&X^{1,\hat{g}}_{t+1}+X^{2,\hat{g}}_{t+1}\nonumber\\
=& X^{1,\hat{g}}_t+X^{2,\hat{g}}_t-D^1_t-D^2_t +A^1_t+A^2_t\nonumber\\
&+ D^1_t1_{\left\{X^1_t=0, S_t = 1\right\}}+D^2_t1_{\left\{X^1_t\neq 0, S_t = 1\right\}}\nonumber\\
&+ D^1_t1_{\left\{ S_t =0\right\}}+D^2_t1_{\left\{S_t =0\right\}} \nonumber\\
=& S_{t+1},
\end{align}
where the last equality follows by the definition of $S_{t+1}$.
\\
Therefore, at any time $t \geq T_0+1$ we have
\begin{align}
 X^{1,\hat{g}}_t+X^{2,\hat{g}}_t =S_t.
\end{align}
The proof of claim (\ref{eq:pfinf:claim}), and consequently, the proof of Theorem \ref{thm:inf} is complete.
\end{proof}

%
%

\begin{acknowledgements}
This work was partially supported by National Science Foundation (NSF) Grant CCF-1111061 and NASA grant NNX12AO54G.
The authors thank Mark Rudelson and Aditya Mahajan for helpful discussions.
\end{acknowledgements}

\bibliographystyle{spbasic}      
\bibliography{decroutingref}

\begin{thebibliography}{32}
\providecommand{\natexlab}[1]{#1}
\providecommand{\url}[1]{{#1}}
\providecommand{\urlprefix}{URL }
\expandafter\ifx\csname urlstyle\endcsname\relax
  \providecommand{\doi}[1]{DOI~\discretionary{}{}{}#1}\else
  \providecommand{\doi}{DOI~\discretionary{}{}{}\begingroup
  \urlstyle{rm}\Url}\fi
\providecommand{\eprint}[2][]{\url{#2}}

\bibitem[{Abdollahi and Khorasani(2008)}]{abdollahi2008novel}
Abdollahi F, Khorasani K (2008) A novel {$H_\infty$} control strategy for
  design of a robust dynamic routing algorithm in traffic networks. IEEE
  Journal on Selected Areas in Communications 26(4):706--718

\bibitem[{Akgun et~al(2012)Akgun, Righter, and Wolff}]{akgun2012understanding}
Akgun OT, Righter R, Wolff R (2012) Understanding the marginal impact of
  customer flexibility. Queueing Systems 71(1-2):5--23

\bibitem[{Aumann(1976)}]{aumann1976agreeing}
Aumann RJ (1976) Agreeing to disagree. The Annals of Statistics pp 1236--1239

\bibitem[{Beutler and Teneketzis(1989)}]{beutler1989routing}
Beutler FJ, Teneketzis D (1989) Routing in queueing networks under imperfect
  information: Stochastic dominance and thresholds. Stochastics and Stochastic
  Reports 26(2):81--100

\bibitem[{Bremaud(1999)}]{bremaud1999markov}
Bremaud P (1999) Markov chains: Gibbs fields, Monte Carlo simulation, and
  queues, vol~31. springer

\bibitem[{Cogill et~al(2006)Cogill, Rotkowitz, Van~Roy, and
  Lall}]{cogill2006approximate}
Cogill R, Rotkowitz M, Van~Roy B, Lall S (2006) An approximate dynamic
  programming approach to decentralized control of stochastic systems. In:
  Control of Uncertain Systems: Modelling, Approximation, and Design, Springer,
  pp 243--256

\bibitem[{Davis(1977)}]{davis1977optimal}
Davis E (1977) Optimal control of arrivals to a two-server queueing system with
  separate queues. PhD thesis, PhD dissertation, Program in Operations
  Research, North Carolina State University, Raleigh, NC

\bibitem[{Ephremides et~al(1980)Ephremides, Varaiya, and
  Walrand}]{ephremides1980simple}
Ephremides A, Varaiya P, Walrand J (1980) A simple dynamic routing problem.
  IEEE Transactions on Automatic Control 25(4):690--693

\bibitem[{Foley and McDonald(2001)}]{foley2001join}
Foley RD, McDonald D (2001) Join the shortest queue: stability and exact
  asymptotics. Annals of Applied Probability pp 569--607

\bibitem[{Hajek(1984)}]{hajek1984optimal}
Hajek B (1984) Optimal control of two interacting service stations. IEEE
  Transactions on Automatic Control 29(6):491--499

\bibitem[{Ho(1980)}]{ho1980team}
Ho YC (1980) Team decision theory and information structures. Proceedings of
  the IEEE 68(6):644--654

\bibitem[{Hordijk and Koole(1990)}]{hordijk1990optimality}
Hordijk A, Koole G (1990) On the optimality of the generalized shortest queue
  policy. Probability in the Engineering and Informational Sciences
  4(4):477--487

\bibitem[{Hordijk and Koole(1992)}]{hordijk1992assignment}
Hordijk A, Koole G (1992) On the assignment of customers to parallel queues.
  Probability in the Engineering and Informational Sciences 6(04):495--511

\bibitem[{Kuri and Kumar(1995)}]{kuri1995optimal}
Kuri J, Kumar A (1995) Optimal control of arrivals to queues with delayed queue
  length information. IEEE Transactions on Automatic Control 40(8):1444--1450

\bibitem[{Lin and Kumar(1984)}]{lin1984optimal}
Lin W, Kumar P (1984) Optimal control of a queueing system with two
  heterogeneous servers. IEEE Transactions on Automatic Control 29(8):696--703

\bibitem[{Mahajan(2013)}]{mahajan2012optimal}
Mahajan A (2013) Optimal decentralized control of coupled subsystems with
  control sharing. IEEE Transactions on Automatic Control 58(9):2377--2382,
  \doi{10.1109/TAC.2013.2251807}

\bibitem[{Manfredi(2014)}]{manfredi2014decentralized}
Manfredi S (2014) Decentralized queue balancing and differentiated service
  scheme based on cooperative control concept. IEEE Transactions on Industrial
  Informatics 10(1):586--593

\bibitem[{Marshall et~al(2010)Marshall, Olkin, and
  Arnold}]{marshall2010inequalities}
Marshall A, Olkin I, Arnold B (2010) Inequalities: theory of majorization and
  its applications. Springer Verlag

\bibitem[{Menich and Serfozo(1991)}]{menich1991optimality}
Menich R, Serfozo RF (1991) Optimality of routing and servicing in dependent
  parallel processing systems. Queueing Systems 9(4):403--418

\bibitem[{Nayyar et~al(2013)Nayyar, Mahajan, and
  Teneketzis}]{nayyar2013decentralized}
Nayyar A, Mahajan A, Teneketzis D (2013) Decentralized stochastic control with
  partial history sharing: A common information approach. IEEE Transactions on
  Automatic Control 58(7):1644--1658, \doi{10.1109/TAC.2013.2239000}

\bibitem[{Ouyang and Teneketzis(2013)}]{ouyang2013}
Ouyang Y, Teneketzis D (2013) A routing problem in a simple queueing system
  with non-classical information structure. In: Proc. 51th Annual Allerton
  Conference on Communication, Control, and Computing (Allerton), Monticello,
  IL, pp 1278 -- 1284

\bibitem[{Ouyang and Teneketzis(2014)}]{ouyang2014}
Ouyang Y, Teneketzis D (2014) Balancing through signaling in decentralized
  routing. In: Proc. 53rd Conference on Decision and Control, Los Angeles, CA,
  accepted

\bibitem[{Pandelis and Teneketzis(1996)}]{pandelis1996simple}
Pandelis DG, Teneketzis D (1996) A simple load balancing problem with
  decentralized information. Mathematical Methods of Operations Research
  44(1):97--113

\bibitem[{Petersen and Petersen(1989)}]{petersen1989ergodic}
Petersen KE, Petersen K (1989) Ergodic theory, vol~2. Cambridge University
  Press

\bibitem[{Reddy et~al(2012)Reddy, Banerjee, Gopalan, Shakkottai, and
  Ying}]{reddy2012distributed}
Reddy AA, Banerjee S, Gopalan A, Shakkottai S, Ying L (2012) On distributed
  scheduling with heterogeneously delayed network-state information. Queueing
  Systems 72(3-4):193--218

\bibitem[{Si et~al(2013)Si, Zhu, Du, and Xie}]{si2013decentralized}
Si X, Zhu XL, Du X, Xie X (2013) A decentralized routing control scheme for
  data communication networks. Mathematical Problems in Engineering 2013,
  article ID 648267

\bibitem[{Weber(1978)}]{weber1978optimal}
Weber RR (1978) On the optimal assignment of customers to parallel servers.
  Journal of Applied Probability pp 406--413

\bibitem[{Weber and Stidham~Jr(1987)}]{weber1987optimal}
Weber RR, Stidham~Jr S (1987) Optimal control of service rates in networks of
  queues. Advances in applied probability pp 202--218

\bibitem[{Whitt(1986)}]{whitt1986deciding}
Whitt W (1986) Deciding which queue to join: Some counterexamples. Operations
  research 34(1):55--62

\bibitem[{Winston(1977)}]{winston1977optimality}
Winston W (1977) Optimality of the shortest line discipline. Journal of Applied
  Probability pp 181--189

\bibitem[{Witsenhausen(1971)}]{witsenhausen1971separation}
Witsenhausen HS (1971) Separation of estimation and control for discrete time
  systems. Proceedings of the IEEE 59(11):1557--1566

\bibitem[{Ying and Shakkottai(2011)}]{ying2011throughput}
Ying L, Shakkottai S (2011) On throughput optimality with delayed network-state
  information. IEEE Transactions on Information Theory 57(8):5116--5132

\end{thebibliography}


\end{document}